\documentclass[11pt]{article}

\usepackage{amsmath,amssymb,amsthm,mathtools}
\usepackage[T1]{fontenc}
\usepackage[expansion=false]{microtype}
\usepackage{graphicx}
\usepackage{epstopdf}
\usepackage{tikz}
\usetikzlibrary{decorations.pathmorphing,calc}
\usepackage{cite}
\usepackage{caption}
\usepackage{subcaption}
\usepackage[a4paper,margin=1.05in]{geometry}
\usepackage{enumitem}
\usepackage{bm}
\usepackage{authblk}
\usepackage[colorlinks=true,linkcolor=blue!55!black,citecolor=blue!55!black,urlcolor=blue!55!black]{hyperref}

\usepackage{abstract}

\theoremstyle{plain}
\newtheorem{theorem}{Theorem}[section]
\newtheorem{lemma}[theorem]{Lemma}
\newtheorem{proposition}[theorem]{Proposition}
\newtheorem{corollary}[theorem]{Corollary}

\theoremstyle{definition}

\newtheorem{remark}[theorem]{Remark}

\newcommand{\R}{\mathbb{R}}
\newcommand{\Z}{\mathbb{Z}}
\newcommand{\N}{\mathbb{N}}

\newcommand{\eps}{\varepsilon}
\newcommand{\dd}{\mathrm{d}}
\newcommand{\Hadamard}{\circ}
\newcommand{\Sp}{S_{+}}
\newcommand{\Sm}{S_{-}}
\newcommand{\norm}[1]{\left\lVert #1 \right\rVert}
\newcommand{\abs}[1]{\left\lvert #1 \right\rvert}
\newcommand{\ip}[2]{\left( #1,\, #2 \right)}

\newcommand{\Hper}{H^{2}_{\mathrm{per}}}

\DeclareMathOperator{\Ker}{Ker}
\DeclareMathOperator{\Coker}{Coker}
\DeclareMathOperator{\spanop}{span}
\newcommand{\notapprox}{\not\approx}
\newcommand{\ord}[1]{{<#1>}}

\newcommand{\fourfreq}{(\omega^{(a)})^2,\dots,(\omega^{(d)})^2}
\newcommand{\optband}{(2k_2,2k_1+2k_2)}
\newcommand{\bothbands}{(0,2k_1)\cup(2k_2,2k_1+2k_2)}

\begin{document}

\title{Generalization of a localized-state formation mechanism\\
in finite lattices with interaction nonlinearity}

\author[1]{Huajie Song\thanks{Email: \texttt{songhj@hust.edu.cn} }}
\author[1]{Haitao Xu\thanks{Corresponding author. Email: \texttt{hxumath@hust.edu.cn} }}
\affil[1]{School of Mathematics and Statistics, Huazhong University of Science and Technology, Wuhan 430074, China}
\date{}

\maketitle

\begin{abstract}
We study how time-periodic, spatially localized states are born from the linear spectrum of a
\emph{finite} lattice as the nonlinearity is switched on. In earlier work we treated this question
for a diatomic chain with on-site nonlinearity and developed a framework that continues a
near-edge linear mode in amplitude and controls the resulting perturbation series uniformly in
the chain length. The present paper shows that the same framework applies to the more difficult
case of Fermi--Pasta--Ulam--Tsingou (FPUT) interaction nonlinearity. The key is a structural
relation between the FPUT and on-site nonlinearities, which allows the estimates obtained in the
on-site setting to be transferred to the FPUT setting. As before, the analysis yields a
quantitative radius of convergence, $\eps=\Theta(1/\sqrt{n})$ for a chain of length $2n$, below
which the near-edge mode stays extended and above which the orbit localizes and its frequency
leaves the band. The diatomic chain is used only as a test case; both the formation mechanism and
the method are model-independent and are expected to extend to other short-range nonlinearities
and to higher dimensions.

\end{abstract}

\section{Introduction}\label{sec:intro}

A central question in nonlinear lattice dynamics is the mechanism by which spatially localized
states form. In most existing work the starting point is already a localized object: discrete
breathers are continued from the anti-continuum limit \cite{MacKayAubry,FlachGorbach}, spatial
dynamics and variational methods construct exponentially localized states directly
\cite{JamesNoble2004,Vainchtein2022}, and topological edge states are continued into the nonlinear
regime \cite{HasanKane,SSH,WangBertoldi,SmirnovaNonlinear,ChaunsaliStrong}. In all of these the
localized character is present from the outset.

The mechanism we study is different: we begin from an \emph{extended} linear mode near a band
edge and ask how the nonlinearity turns it into a localized state. As the amplitude grows the
mode's frequency drifts, and once it leaves the band the state is no longer extended. The object
of interest is this transition and the amplitude at which it occurs. This question is meaningful
only for a finite system. On an infinite lattice the bands are continuous, there is no isolated
near-edge frequency to continue, and the question cannot even be posed; much existing work also
relies on idealizations such as the infinite or periodic lattice and the continuum approximation
\cite{HadadKhanikaev,LeykamChong}. In a finite system the near-edge frequency is isolated and sits
at a positive distance from the band edge, so continuing it is a well-posed small-divisor problem,
and the natural quantity is the radius of convergence of the amplitude expansion: how far the
extended mode can be continued before it ceases to be extended.

Our interest is in this mechanism rather than in any one model. To make it concrete we work with a
finite one-dimensional diatomic chain, which we regard as a test case and not as a restriction; we
expect the mechanism and the method to apply to other short-range nonlinearities and to higher
dimensions. In earlier work \cite{SongXuArxiv}, building on the linear spectral analysis of
\cite{SongXuJPA}, we treated the chain with on-site (substrate) nonlinearity, branching a
time-periodic orbit off a near-edge mode by a Lyapunov--Schmidt reduction and the implicit
function theorem \cite{Kielhofer,Kuznetsov} and controlling the amplitude expansion uniformly in
the chain length $2n$. The analysis is perturbative: it controls the expansion up to its radius of
convergence, which is itself the quantity of interest and already captures the onset of
localization, rather than following the branch deep into the gap. A central conclusion was that
the radius of convergence must be estimated at the same order as the spectral gaps that control
it; the two scales coincide, both being $\Theta(1/\sqrt n)$.

The present paper shows that this framework is not tied to the on-site nonlinearity. We replace
the on-site cubic by an FPUT cubic \cite{FPUT}, which acts on the relative displacement of
neighbouring particles, that is, on the strain of each spring; this interaction nonlinearity is
the basic model of anharmonic forces in mechanical, acoustic, granular, and metamaterial lattices.
It is also the harder of the two cases, because the FPUT force couples each particle to its
neighbours through differences, so the quantities to be estimated are products of \emph{differences}
of eigenmodes rather than of single eigenmodes. That the framework nonetheless applies rests on a
structural relation between the two nonlinearities, made precise in Section~\ref{sec:middle}:
the inner products that enter the perturbation analysis agree for the two problems up to a small,
controllable error, so the cancellation estimates built for the on-site problem can be reused. The
same relation shows that a single boundary configuration supports two coexisting localized
families, one concentrated in the middle of the chain and one at the boundary, selected by an
arithmetic condition on the indices. We anchor the continuation at the lower optical edge for
definiteness; the construction needs only an isolated near-edge frequency at a positive
$\Theta(1/n^2)$ distance from a band edge, which holds at every edge of the finite chain.

The paper is organized as follows. Section~\ref{sec:prelim} fixes the model and the linear
spectral facts and introduces the objects $\Sp,\Sm,M$ used for the FPUT nonlinearity.
Section~\ref{sec:continuation} carries out the Lyapunov--Schmidt reduction and sets up the
perturbation hierarchy. Sections~\ref{sec:middle} and~\ref{sec:boundary} establish the
middle-localized and edge-localized families. Section~\ref{sec:conclusion} states our conclusions,
and the technical proofs are collected in the appendices.

\section{The model and preliminaries}\label{sec:prelim}

\subsection{The finite FPUT-type diatomic chain}
We consider a chain of $2n$ identical unit masses indexed $j=1,\dots,2n$, connected by
alternating interior springs of stiffness $k_1$ and $k_2$, and clamped at both ends by boundary
springs of stiffness $k_{3,1}$ (left) and $k_{3,2}$ (right) to fixed walls $q_0=q_{2n+1}=0$.
Each spring carries, in addition to its linear (harmonic) force, a cubic restoring force that
depends on the relative displacement of its two endpoints. The equations of motion are
\begin{equation}\label{eq:fput}
\begin{aligned}
\ddot q_{1}&=k_{3,1}(0-q_1)+k_1(q_2-q_1)-(0-q_1)^3-(q_2-q_1)^3,\\
\ddot q_{2m}&=k_1(q_{2m-1}-q_{2m})+k_2(q_{2m+1}-q_{2m})-(q_{2m-1}-q_{2m})^3-(q_{2m+1}-q_{2m})^3,\\
\ddot q_{2m+1}&=k_2(q_{2m}-q_{2m+1})+k_1(q_{2m+2}-q_{2m+1})-(q_{2m}-q_{2m+1})^3-(q_{2m+2}-q_{2m+1})^3,\\
\ddot q_{2n}&=k_1(q_{2n-1}-q_{2n})+k_{3,2}(0-q_{2n})-(q_{2n-1}-q_{2n})^3-(0-q_{2n})^3,
\end{aligned}
\end{equation}
for $1\le m\le n-1$ (see Figure~\ref{fig:schematic}). In vector form,
\begin{equation}\label{eq:fput-matrix}
\ddot q = Lq+N(q),\qquad q=(q_1,\dots,q_{2n})^{\top},
\end{equation}
where the symmetric tridiagonal matrix $L$ is the linear dynamical operator of the clamped
diatomic chain and $N(q)$ collects the cubic interparticle (FPUT) forces. The conserved energy is
\begin{equation}\label{eq:energy}
\begin{aligned}
E={}&\sum_{j=1}^{2n}\tfrac12\dot q_j^{2}
+\sum_{j=1}^{n}\tfrac{k_1}{2}(q_{2j-1}-q_{2j})^{2}
+\sum_{j=1}^{n-1}\tfrac{k_2}{2}(q_{2j+1}-q_{2j})^{2}
+\tfrac{k_{3,1}}{2}q_1^{2}+\tfrac{k_{3,2}}{2}q_{2n}^{2}\\
&-\sum_{j=1}^{n}\tfrac14(q_{2j-1}-q_{2j})^{4}
-\sum_{j=1}^{n-1}\tfrac14(q_{2j+1}-q_{2j})^{4}
-\tfrac14 q_1^{4}-\tfrac14 q_{2n}^{4}.
\end{aligned}
\end{equation}

\begin{figure}[t]
\centering
\includegraphics[width=1\linewidth]{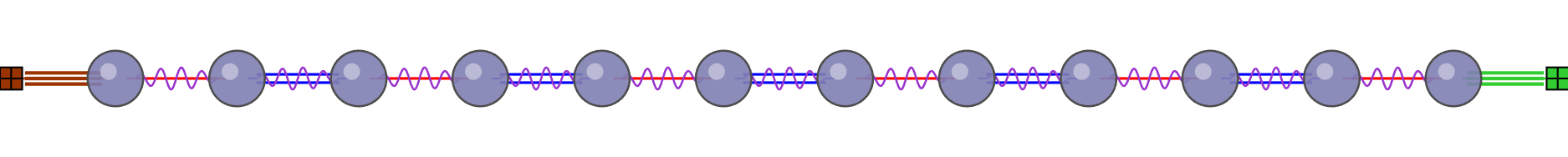}
\caption{Schematic of the finite FPUT-type diatomic chain \eqref{eq:fput} of length $2n$. The
masses are identical; the straight single and double lines denote the linear interior springs
$k_1,k_2$, the springs at the two ends denote the clamped boundary springs $k_{3,1},k_{3,2}$ to
fixed walls, and the curved springs denote the cubic FPUT nonlinearity, which acts on the
relative displacement of each pair of connected sites.}
\label{fig:schematic}
\end{figure}

We seek time-periodic solutions. Setting $\tau=\omega t$ and $Q(\tau)=q(\tau/\omega)=q(t)$, the
chain \eqref{eq:fput} becomes
\begin{equation}\label{eq:fput-rescaled}
\omega^2\,\frac{\dd^2}{\dd\tau^2}Q(\tau)=LQ(\tau)+N(Q(\tau)),
\qquad Q\in\Hper(\R/2\pi\Z\times\R^{2n}).
\end{equation}
Throughout we use the asymptotic notation and conventions of \cite{SongXuJPA,SongXuArxiv}: in
particular $A\lesssim B$ means $A\le CB$ and $A=\Theta(B)$ means $c_1B\le A\le c_2 B$, for
positive constants independent of the chain length $n$. We assume $k_1<k_2$ and work in the
\emph{long-chain regime}: $n$ exceeds a threshold $n_0$ determined by the model parameters, so
that the decay factors of the edge states are negligible (in the numerical examples $n\ge50$
suffices).

\subsection{Linear spectral facts}\label{sec:linear}
Setting the nonlinearity to zero in \eqref{eq:fput-matrix} leaves the linear chain
$\ddot q=Lq$. Looking for normal modes $q(t)=u\,e^{i\omega t}$ turns this into the eigenvalue
problem
\begin{equation}\label{eq:eigenproblem}
Lu^{(j)}=-(\omega^{(j)})^2u^{(j)},\qquad 1\le j\le 2n,
\end{equation}
where $-(\omega^{(j)})^2$ are the eigenvalues of the symmetric matrix $L$ and $u^{(j)}$ the
corresponding eigenvectors, which we take real and in the standard form fixed below. The spectral
theory of $L$ in the long-chain limit was developed in \cite{SongXuJPA} (for the matrix-analytic
background see \cite{HornJohnson,Parlett}); we recall only what is needed. The squared
frequencies $(\omega^{(j)})^2$ fill two bands,
\[
\text{acoustic band }(0,2k_1),\qquad\text{optical band }(2k_2,2k_1+2k_2),
\]
separated by the gap $(2k_1,2k_2)$, and in addition there are at most two isolated frequencies
lying outside the bands, which correspond to localized (edge) states \cite{SongXuJPA}. We carry
out the continuation from a frequency just above the lower edge $2k_2$ of the optical band; as
explained in Section~\ref{sec:intro}, this is a representative choice, and the same construction
applies at any band edge. Following the convention of \cite{SongXuJPA}, we order the optical
near-edge frequencies upward from the lower optical edge,
\begin{equation}\label{eq:ordering}
2k_2<(\omega^{(1)})^2<(\omega^{(2)})^2<\cdots<(\omega^{(n_1)})^2<2k_1+2k_2 ,
\end{equation}
and the acoustic frequencies downward from the upper acoustic edge,
\begin{equation}\label{eq:ordering-ac}
0<(\omega^{(2n)})^2<(\omega^{(2n-1)})^2<\cdots<(\omega^{(2n-n_2+1)})^2<2k_1,
\end{equation}
so that $n_1+n_2$ together with the (at most two) isolated frequencies accounts for all $2n$
modes. The near-edge optical frequencies in \eqref{eq:ordering} are the ones used in the
continuation; the acoustic frequencies \eqref{eq:ordering-ac} and the isolated frequencies also
enter the estimates below.

\smallskip\noindent\textbf{Eigenmode parametrization.}
For frequencies near the lower optical edge, the eigenvector $u^{(k)}$ admits the explicit
trigonometric form (cf.\ \cite{SongXuJPA})
\begin{equation}\label{eq:eigform}
\begin{pmatrix}u^{(k)}_{2j-1}\\[1mm]u^{(k)}_{2j}\end{pmatrix}
=(-1)^{j-1}
\begin{pmatrix}\sin\!\big(\Delta\alpha^{(k)}+\Delta\beta^{(k)}-(j-1)\Delta\theta^{(k)}\big)\\[1mm]
\sin\!\big(-\Delta\alpha^{(k)}+\Delta\beta^{(k)}-(j-1)\Delta\theta^{(k)}\big)\end{pmatrix},
\end{equation}
where the small parameters $\Delta\theta^{(k)},\Delta\alpha^{(k)},\Delta\beta^{(k)}$ measure the
distance from the band edge. In the two boundary regimes relevant below,
\begin{equation}\label{eq:dtheta}
\Delta\theta^{(k)}\approx\frac{k\pi}{n-1}\quad\text{(general right boundary)},\qquad
\Delta\theta^{(k)}\approx\frac{(2k-1)\pi}{2(n-1)}\quad\text{(special right boundary)},
\end{equation}
for $k\ll n$, with $\Delta\alpha^{(k)},\Delta\beta^{(k)}=\Theta(k/n)$. In particular
$\abs{\Delta\theta^{(k)}}=\Theta(k/n)$ for $k\ll n$, and the near-edge gaps obey
\begin{equation}\label{eq:gap}
(\omega^{(k)})^2-(\omega^{(1)})^2\approx\frac{k_1k_2(k^2-1)\pi^2}{2(k_2-k_1)(n-1)^2}
=\Theta\!\Big(\frac{k^2}{n^2}\Big),
\end{equation}
which is the quantitative statement that the optical band becomes a quasi-continuum of width
$\Theta(1/n^2)$ at its edge. The leading mode $u^{(1)}$ is a near-edge \emph{extended} state, not
an edge state; this is the object we continue into the nonlinear regime.

\smallskip\noindent\textbf{Norm and inner-product estimates.}
We use the following facts, established for the present boundary regimes in \cite{SongXuJPA}.
Here and below the eigenvectors $u^{(j)}$ are taken in the standard trigonometric form
\eqref{eq:eigform}, in which the leading coefficient equals $1$. For an in-band index
($(\omega^{(k)})^2$ in the acoustic or optical band) and any $\eps>0$,
\begin{equation}\label{eq:normest}
\big|\,\abs{u^{(k)}}^2-n\,\big|\lesssim n^{\eps},
\end{equation}
so that in this scaling the in-band eigenvectors have squared norm $\Theta(n)$. The
estimate \eqref{eq:normest} is the worst case over the whole band; for frequencies close to the
band edge, that is for $k\ll n$, the deviation is in fact smaller, and the error term in
\eqref{eq:normest} can be replaced by an $O(1)$ bound. Since the continuation is anchored at such
near-edge frequencies, we use this sharper form where it is needed. Where a unit-length vector is
required we divide by $\abs{u^{(k)}}$ explicitly. The Hadamard product is
$(\,x\Hadamard y\,)_i=x_iy_i$. With $u^{(j)}$ in this standard form, the cubic inner products of
in-band eigenvectors obey the following case distinction.

\begin{lemma}[Hadamard cubic inner products, {\cite[Lemma~4.5]{SongXuArxiv}}]\label{lem:hadamard}
Let $u^{(j)}$ ($1\le j\le 2n$) be the eigenvectors of $L$ associated with the eigenvalues
$-(\omega^{(j)})^2$ (these comprise the extended states of the form \eqref{eq:eigform} in both the
acoustic and optical bands, together with the at most two edge states), and let
$a,b,c,d\in\{1,\dots,2n\}$.
\begin{enumerate}[label=\textup{(\roman*)},leftmargin=1.6em]
\item If $(\omega^{(a)})^2,(\omega^{(b)})^2,(\omega^{(c)})^2,(\omega^{(d)})^2\in
(0,2k_1)\cup(2k_2,2k_1+2k_2)$, then
$\bigl|\ip{u^{(a)}u^{(b)}u^{(c)}}{u^{(d)}}\bigr|\le 2n$.
\item If at least one of the four squared frequencies lies outside $[0,2k_1]\cup[2k_2,2k_1+2k_2]$,
then $\bigl|\ip{u^{(a)}u^{(b)}u^{(c)}}{u^{(d)}}\bigr|\le\Theta(1)$.
\item If $k:=\max\{a,b,c,d\}\ll n$ and the four squared frequencies all lie in
$(2k_2,2k_1+2k_2)$, then for $k_{3,1}=\Theta(1)\notapprox2k_2$, $k_{3,2}=\Theta(1)\notapprox2k_2$, calling the
quadruple $(a,b,c,d)$ \emph{balanced} when some choice of signs gives $a\pm b\pm c\pm d=0$,
\[
\bigl|\ip{u^{(a)}u^{(b)}u^{(c)}}{u^{(d)}}\bigr|=
\begin{cases}
\Theta(1), & (a,b,c,d)\ \text{not balanced},\\[1mm]
\Theta(n), & (a,b,c,d)\ \text{balanced};
\end{cases}
\]
\item If the four squared frequencies lie in $(2k_2,2k_1+2k_2)$ and
$\max\{a,b,c\}\ll n^{1-\eps}<d<n_1-a-b-c$, then
$\ip{u^{(a)}u^{(b)}u^{(c)}}{u^{(d)}}\lesssim n^{2\eps}$.
\end{enumerate}
\end{lemma}

The interpretation of (iii) is essential: the on-site cubic inner product is generically
$\Theta(1)$, except on the resonant set where the four indices satisfy a parity-balanced linear
relation, where it jumps to $\Theta(n)$. These resonances are exactly the small divisors that
control the nonlinear continuation.

\subsection{The shift operators and the strain trilinear form}
The FPUT force at a site depends on how much the two springs attached to that site are stretched,
that is, on the relative displacements of neighbouring particles. To write this compactly we use
the shift operators $\Sp,\Sm$ on $\R^{2n}$ defined by
\begin{equation}\label{eq:shift}
(\Sp v)_j=v_{j+1},\qquad (\Sm v)_j=v_{j-1},\qquad v_0=v_{2n+1}=0,
\end{equation}
so that $(\Sp v-v)_j=v_{j+1}-v_j$ and $(v-\Sm v)_j=v_j-v_{j-1}$ are the strains of the right and
left springs at site $j$. We write $x\Hadamard y$ for the Hadamard (componentwise) product,
$(x\Hadamard y)_i=x_iy_i$, and abbreviate the threefold Hadamard product as
$x^{\Hadamard3}=x\Hadamard x\Hadamard x$, i.e.\ $(x^{\Hadamard3})_i=x_i^3$. With this notation the
cubic FPUT force \eqref{eq:fput} reads, in vector form,
\begin{equation}\label{eq:Nform}
N(Q)=(Q-\Sm Q)^{\Hadamard3}-(\Sp Q-Q)^{\Hadamard3},
\end{equation}
the difference between the cubed strains of the left and right springs at each site.

In the on-site problem the corresponding cubic term was the Hadamard cube $u^{\Hadamard3}$, and
the key estimates concerned inner products of the form
$\ip{u^{(a_1)}\Hadamard u^{(a_2)}\Hadamard u^{(a_3)}}{u^{(a_4)}}$. For the FPUT force the cube is
applied to the \emph{strains} rather than to the displacements, so the relevant object is the
symmetric trilinear form
\begin{equation}\label{eq:Mdef}
\begin{aligned}
M\!\left(u^{(a_1)},u^{(a_2)},u^{(a_3)}\right)
&=\big(u^{(a_1)}-\Sm u^{(a_1)}\big)\Hadamard
   \big(u^{(a_2)}-\Sm u^{(a_2)}\big)\Hadamard
   \big(u^{(a_3)}-\Sm u^{(a_3)}\big)\\
&\quad-\big(\Sp u^{(a_1)}-u^{(a_1)}\big)\Hadamard
   \big(\Sp u^{(a_2)}-u^{(a_2)}\big)\Hadamard
   \big(\Sp u^{(a_3)}-u^{(a_3)}\big),
\end{aligned}
\end{equation}
of which \eqref{eq:Nform} is the special case with all three arguments equal: $M(u,u,u)=N(u)$.
The FPUT analysis thus requires estimates of
$\ip{M(u^{(a_1)},u^{(a_2)},u^{(a_3)})}{u^{(a_4)}}$ in place of the on-site Hadamard inner products.
Lemma~\ref{lem:M-vs-hadamard} below shows that these two quantities are closely related, which is
the structural fact that makes the FPUT problem tractable.

\begin{remark}[Why the strains do not degenerate]\label{rem:edge}
The trilinear form \eqref{eq:Mdef} is built entirely from the discrete strains
$u^{(k)}-\Sm u^{(k)}$ and $\Sp u^{(k)}-u^{(k)}$, so the whole FPUT analysis would collapse if
these strains were small: $M$ would then be negligible and the reduction below would say nothing.
The point of this remark is that for the modes we use this does not happen. Near a band edge the
two sublattices of the diatomic cell move out of phase, so neighbouring displacements differ at
leading order and the strains have $\Theta(1)$ leading components; hence \eqref{eq:Mdef} is
non-degenerate. The precise componentwise sizes are recorded in
\eqref{eq:strain-odd}--\eqref{eq:strain-even} and are the quantitative input for the estimates of
Appendix~\ref{app:reduction}. Only this non-degeneracy is needed, and it holds at every band edge
except the long-wavelength acoustic edge ($\omega^2\to0$), where neighbouring sites move nearly in
phase and the strains vanish; in particular it holds at the lower optical edge we continue from,
so the construction is not tied to that particular choice (cf.\ Section~\ref{sec:intro}).
\end{remark}

\section{Lyapunov--Schmidt reduction and the perturbation hierarchy}\label{sec:continuation}

We branch off the linear extended state $u^{(1)}$ at the lower optical edge. Because
$2\omega^{(1)}\approx2\sqrt{2k_2}>\sqrt{2k_1+2k_2}$, the non-resonance condition
$\omega^{(j)}/\omega^{(1)}\notin\Z$ holds for $2\le j\le2n$, and the implicit function theorem
applies.

\begin{lemma}[Existence of the periodic branch]\label{lem:existence}
Near $(\omega,Q)=(\omega^{(1)},0)$ the nonlinear system \eqref{eq:fput-rescaled} possesses a
one-parameter family of time-periodic solutions.
\end{lemma}

\begin{proof}[Proof sketch]
Define $g(\omega,Q(\tau))=\big(\omega^2\tfrac{\dd^2}{\dd\tau^2}-L\big)Q(\tau)-N(Q(\tau))$ on
$\Hper(\R/2\pi\Z\times\R^{2n})$, so that $g(\omega,0)=0$. The linearization at the trivial branch,
$(\omega^{(1)})^2\tfrac{\dd^2}{\dd\tau^2}-L$, has, by the non-resonance condition
$\omega^{(j)}/\omega^{(1)}\notin\Z$,
\begin{equation}\label{eq:kernel}
\Ker\!\Big((\omega^{(1)})^2\tfrac{\dd^2}{\dd\tau^2}-L\Big)
=\spanop\{u^{(1)}e^{i\tau},u^{(1)}e^{-i\tau}\}
=\Coker\!\Big((\omega^{(1)})^2\tfrac{\dd^2}{\dd\tau^2}-L\Big).
\end{equation}
A Lyapunov--Schmidt reduction with respect to this two-dimensional kernel, using the translation
($\tau\mapsto\tau+\eta$) and reflection ($\tau\mapsto-\tau$) symmetries of
\eqref{eq:fput-rescaled} to eliminate the phase, reduces the problem to a single scalar
bifurcation equation $g_4(\omega,\abs{A})=0$ for the amplitude $\abs A$, with
$g_4(\omega^{(1)},0)=0$ and $\partial_\omega g_4|_{(\omega^{(1)},0)}\ne0$. The implicit function
theorem then yields $\omega=\omega(\abs A)$ and a corresponding solution in a neighbourhood of
$(\omega^{(1)},0)$. The full reduction, including the explicit form of the projected nonlinearity,
is carried out in Appendix~\ref{sec:proofs}.
\end{proof}

By analyticity, the branch admits the convergent-near-zero expansions
\begin{equation}\label{eq:expansion}
\begin{aligned}
\omega&=\omega^{\ord{0}}+\eps\,\omega^{\ord{1}}+\eps^2\omega^{\ord{2}}+\cdots,\\
Q(\tau)&=\eps\,Q^{\ord{0}}(\tau)+\eps^2Q^{\ord{1}}(\tau)+\eps^3Q^{\ord{2}}(\tau)+\cdots,
\end{aligned}
\end{equation}
with $\omega^{\ord{0}}=\omega^{(1)}$ and
$Q^{\ord{0}}(\tau)=\frac{u^{(1)}}{\abs{u^{(1)}}}(e^{i\tau}+e^{-i\tau})$. The implicit
function theorem guarantees \eqref{eq:expansion} for some $\eps$ small but provides \emph{no}
radius of convergence; supplying that radius, uniformly in $n$, is the content of the main
theorem. We first record the parity structure of the hierarchy.

\begin{proposition}[Odd orders vanish]\label{prop:odd}
For all $m\ge0$, $\ \omega^{\ord{2m+1}}=0$ and $Q^{\ord{2m+1}}(\tau)=0$, and
\begin{equation}\label{eq:Qeven}
Q^{\ord{2m}}(\tau)=\sum_{j=2}^{2n}c_{2m,1,j}\frac{u^{(j)}}{\abs{u^{(j)}}}(e^{i\tau}+e^{-i\tau})
+\sum_{k=1}^{m}\sum_{j=1}^{2n}c_{2m,2k+1,j}\frac{u^{(j)}}{\abs{u^{(j)}}}\big(e^{(2k+1)i\tau}+e^{-(2k+1)i\tau}\big).
\end{equation}
\end{proposition}

\begin{proof}
The cubic nonlinearity is odd, so the recursion couples only odd temporal harmonics and even
powers of $\eps$; the proof is the induction of Proposition~\ref{prop:Q2} below, applied at each
order, and is identical in structure to the on-site case \cite{SongXuArxiv}.
\end{proof}

We can now state the main result. Throughout we impose on the left boundary spring the
genericity condition
\begin{equation}\label{eq:leftcond}
\abs{k_{3,1}-k_2}\notapprox k_2,\qquad k_{3,1}<2k_2,
\end{equation}
which excludes the degenerate value at which the left end would itself carry an edge state. The
two regimes are then distinguished only by the choice of the \emph{right} boundary spring
$k_{3,2}$:
\begin{equation}\label{eq:regimes}
\begin{aligned}
\textbf{(R1)}\ \text{(general right spring):}&\quad \abs{k_{3,2}-k_2}\notapprox k_2,\quad k_{3,2}<2k_2;\\[1mm]
\textbf{(R2)}\ \text{(right spring tuned to the edge):}&\quad \abs{k_{3,2}-2k_2}\ll n^{-4},\quad k_{3,2}\le2k_2.
\end{aligned}
\end{equation}
In (R1) the right spring is generic, as in \eqref{eq:leftcond}; in (R2) it is tuned to the band
edge $2k_2$, which moves the index of the dominant second-order correction to the boundary and
produces boundary localization.

\begin{theorem}[Convergence radius and frequency tendency]\label{thm:main}
Assume the long-chain regime and either of the boundary regimes \textup{(R1)} or \textup{(R2)}
in \eqref{eq:regimes}. Then there is a constant $f=\Theta(1)$ such that the expansions
\eqref{eq:expansion} for $Q$ and $\omega$ converge for
\begin{equation}\label{eq:radius}
0<\eps^2<\frac{1}{fn}, \qquad\text{i.e.}\qquad \eps=\Theta(1/\sqrt n).
\end{equation}
Within this range the leading frequency shift is negative, $\omega^{\ord{2}}=\Theta(1/n)<0$, so
the frequency moves monotonically downward, towards the lower optical edge. At the upper end of
the convergence range the cumulative shift is of the same order as the distance from
$(\omega^{(1)})^2$ to the band edge, so to leading order the frequency reaches the edge precisely
as the expansion ceases to converge. Correspondingly the spatial profile concentrates as $\eps$
approaches the threshold, in regime \textup{(R1)} towards the middle of the chain and in regime
\textup{(R2)} towards the boundary.
\end{theorem}

\begin{remark}
The theorem is an order-of-magnitude statement, not a proof that the frequency leaves the band or
that a genuine localized state is reached. The perturbation series controls the branch only up to
its radius of convergence, and the leading shift only shows that the frequency \emph{reaches} the
edge at that order; whether it detaches beyond the edge is outside the perturbative range. What
the result establishes is a definite downward tendency of the frequency and a growing
concentration of the profile, with quantitative thresholds, which is the perturbative signature of
the onset of localization rather than a global existence statement for an edge state.
\end{remark}

The constant $f$ differs between the two regimes; we compute it in
Sections~\ref{sec:middle}--\ref{sec:boundary}. The estimate $\eps=\Theta(1/\sqrt n)$ matches the
square root of the near-edge gap \eqref{eq:gap}: the series converges up to the amplitude at
which the leading frequency shift carries the frequency to the band edge and the smallest divisor
vanishes.

\section{The middle-localized family (regime R1)}\label{sec:middle}

\subsection{Reducing the FPUT inner products to the on-site ones}
The technical heart of the FPUT analysis is that, tested against in-band modes, the trilinear
form $M$ of \eqref{eq:Mdef} differs from $8$ times the on-site Hadamard cube only by a sharply
controlled error. This is what lets the delicate cancellation estimates of \cite{SongXuArxiv} be
carried over directly.

\begin{lemma}[FPUT form versus on-site cube]\label{lem:M-vs-hadamard}
Suppose $\abs{k_{3,1}-k_2}\notapprox k_2$, $\abs{k_{3,2}-k_2}\notapprox k_2$, and let
$a_1,a_2,a_3,a_4\in\N^{*}$ with $a_1,a_2,a_3,a_4\ll n$. Put $k=\max\{a_1,a_2,a_3,a_4\}$. Then
\begin{equation}\label{eq:M-vs-hadamard}
\Bigl|\ip{M(u^{(a_1)},u^{(a_2)},u^{(a_3)})}{u^{(a_4)}}
-8\,\ip{u^{(a_1)}\Hadamard u^{(a_2)}\Hadamard u^{(a_3)}}{u^{(a_4)}}\Bigr|
=O\!\Big(\frac{k^3}{n^2}\Big).
\end{equation}
\end{lemma}

The proof is given in Appendix~\ref{app:reduction}. Because the right-hand side of
\eqref{eq:M-vs-hadamard} is small, the FPUT inner product inherits the size estimates that
Lemma~\ref{lem:hadamard} gives for the on-site cube. In particular it inherits the same case
distinction: the inner product is generically of order $O(1)$ (or smaller), and becomes large,
of order $\Theta(n)$, only on the resonant set where the four indices satisfy a parity-balanced
linear relation. We record this explicitly.

\begin{lemma}[FPUT cubic inner products]\label{lem:M-dichotomy}
Under $\abs{k_{3,1}-k_2}\notapprox k_2$, $\abs{k_{3,2}-k_2}\notapprox k_2$
and $a_1,a_2,a_3\in\N^{*}$ with $a_1,a_2,a_3\ll n$, with $(a,b,c,d)$ balanced as in
Lemma~\ref{lem:hadamard}:
\begin{enumerate}[label=\textup{(\arabic*)},leftmargin=1.6em]
\item If $\fourfreq\in\bothbands$, then
$\bigl|\ip{M(u^{(a)},u^{(b)},u^{(c)})}{u^{(d)}}\bigr|\le 32n$.
\item If at least one of $\fourfreq$ lies outside the bands, then
$\ip{M(u^{(a)},u^{(b)},u^{(c)})}{u^{(d)}}=O(1)$.
\item If $\fourfreq\in\optband$, with $k=\max\{a,b,c,d\}$,
$1\le k\ll n$, $k_{3,1}=\Theta(1)\notapprox2k_2$ and $k_{3,2}=\Theta(1)\notapprox2k_2$, then
\[
\ip{M(u^{(a)},u^{(b)},u^{(c)})}{u^{(d)}}=
\begin{cases}
O\!\big(1+k^3/n^2\big), & (a,b,c,d)\ \text{not balanced},\\[1mm]
O(n), & (a,b,c,d)\ \text{balanced}.
\end{cases}
\]
\item If the four squared frequencies lie in $(2k_2,2k_1+2k_2)$ and
$\max\{a,b,c\}\ll n^{1-\eps}<d<n_1-n^{1-\eps}$, then
$\ip{M(u^{(a)},u^{(b)},u^{(c)})}{u^{(d)}}\lesssim n^{2\eps}$.
\end{enumerate}
\end{lemma}

\begin{proof}
For (1), (3) and (4) the on-site quantity is bounded by Lemma~\ref{lem:hadamard} and the FPUT
correction in \eqref{eq:M-vs-hadamard} is $O(k^3/n^2)$. On the balanced (resonant) set of (3) the
$\Theta(n)$ on-site value dominates the correction, and the factor $8$ turns the on-site bound $2n$
of Lemma~\ref{lem:hadamard}(i) into $16n\le32n$; off the resonant set in (3) the on-site value is
only $\Theta(1)$, so the correction is no longer dominated and the FPUT inner product is
$O(1+k^3/n^2)$. For (2), at least one out-of-band index
forces the on-site value to $O(1)$ by Lemma~\ref{lem:hadamard}(ii) while the trilinear form $M$ is built
from the same eigenvectors, whose out-of-band exponential decay yields the same $O(1)$ bound
directly.
\end{proof}

\subsection{Second-order correction and the frequency shift}
Collecting the $\Theta(\eps^3)$ terms in \eqref{eq:fput-rescaled} gives the equation for
$Q^{\ord{2}}$,
\begin{equation}\label{eq:order3}
\begin{aligned}
\Big((\omega^{\ord{0}})^2\frac{\dd^2}{\dd\tau^2}-L\Big)Q^{\ord{2}}(\tau)
&+2\omega^{\ord{0}}\omega^{\ord{2}}\frac{\dd^2}{\dd\tau^2}Q^{\ord{0}}(\tau)\\
&=\frac{3\,M(u^{(1)},u^{(1)},u^{(1)})}{\abs{u^{(1)}}^3}(e^{i\tau}+e^{-i\tau})
+\frac{M(u^{(1)},u^{(1)},u^{(1)})}{\abs{u^{(1)}}^3}(e^{3i\tau}+e^{-3i\tau}).
\end{aligned}
\end{equation}
Projecting onto $\spanop\{u^{(j)}(e^{i\tau}+e^{-i\tau})\}$ and
$\spanop\{u^{(j)}(e^{3i\tau}+e^{-3i\tau})\}$ yields the coefficients
\begin{equation}\label{eq:c213-c233}
\begin{aligned}
c_{2,1,j}&=\frac{1}{\abs{u^{(1)}}^3\abs{u^{(j)}}}\,
\frac{3\,\ip{M(u^{(1)},u^{(1)},u^{(1)})}{u^{(j)}}}{(\omega^{(j)})^2-(\omega^{(1)})^2},\\
c_{2,3,j}&=\frac{1}{\abs{u^{(1)}}^3\abs{u^{(j)}}}\,
\frac{\ip{M(u^{(1)},u^{(1)},u^{(1)})}{u^{(j)}}}{(\omega^{(j)})^2-3^2(\omega^{(1)})^2}.
\end{aligned}
\end{equation}

\begin{proposition}[Estimate of $Q^{\ord{2}}$, regime R1]\label{prop:Q2}
In regime \textup{(R1)},
\begin{equation}\label{eq:Q2form}
Q^{\ord{2}}(\tau)=\sum_{j=2}^{2n}c_{2,1,j}\frac{u^{(j)}}{\abs{u^{(j)}}}(e^{i\tau}+e^{-i\tau})
+\sum_{j=1}^{2n}c_{2,3,j}\frac{u^{(j)}}{\abs{u^{(j)}}}(e^{3i\tau}+e^{-3i\tau}),
\end{equation}
where the first-harmonic coefficients satisfy
\begin{equation}\label{eq:c21-estimates}
\Bigl|c_{2,1,3}+\tfrac{3(k_2-k_1)n}{2k_1k_2\pi^2}\Bigr|=O(1),
\end{equation}
and, for $j\ne3$,
\begin{equation}\label{eq:c21-estimates-rest}
|c_{2,1,j}|=
\begin{cases}
O\!\big(j^{-2}\big), & j\in[2,n^{1-\eps}]\setminus\{3\},\\[1mm]
O\!\big(n^{-(2-4\eps)}\big), & j\in(n^{1-\eps},\,n_1-n^{1-\eps}],\\[1mm]
O\!\big(n^{-1}\big), & j>n_1-n^{1-\eps},\ (\omega^{(j)})^2\in(0,2k_1)\cup(2k_2,2k_1+2k_2),\\[1mm]
O\!\big(n^{-3/2}\big), & (\omega^{(j)})^2\in(2k_1,2k_2)\cup(2k_1+2k_2,+\infty),
\end{cases}
\end{equation}
and the third-harmonic coefficients satisfy
\begin{equation}\label{eq:c23-estimates}
\Bigl|c_{2,3,1}+\tfrac{3}{8k_2 n}\Bigr|=O\!\big(n^{-2}\big),\qquad
\Bigl|c_{2,3,3}-\tfrac{1}{8k_2 n}\Bigr|=O\!\big(n^{-2}\big),
\end{equation}
together with, for the remaining indices,
\begin{equation}\label{eq:c23-estimates-rest}
|c_{2,3,j}|=
\begin{cases}
O\!\big(n^{-1}\big), & (\omega^{(j)})^2\in(0,2k_1)\cup(2k_2,2k_1+2k_2),\\[1mm]
O\!\big(n^{-3/2}\big), & (\omega^{(j)})^2\in(2k_1,2k_2)\cup(2k_1+2k_2,+\infty).
\end{cases}
\end{equation}
\end{proposition}

\begin{proof}
The form \eqref{eq:Q2form} and \eqref{eq:c213-c233} follow from \eqref{eq:order3}. To estimate the
coefficients we substitute the inner-product estimates of Lemma~\ref{lem:M-dichotomy}, the norm
estimate \eqref{eq:normest}, and the gap estimate \eqref{eq:gap} into the coefficient formulas
\eqref{eq:c213-c233}. The relevant facts are:
\begin{itemize}[leftmargin=1.6em]
\item $\dfrac{\ip{M(u^{(1)},u^{(1)},u^{(1)})}{u^{(3)}}}{\abs{u^{(1)}}^3\abs{u^{(3)}}}
\approx\dfrac{-2n}{n^2}=-\dfrac2n$, and
$\dfrac{\ip{M(u^{(1)},u^{(1)},u^{(1)})}{u^{(1)}}}{\abs{u^{(1)}}^4}\approx\dfrac{6n}{n^2}=\dfrac6n$
(the index $j=3$ realizes the large $\Theta(n)$ case of Lemma~\ref{lem:M-dichotomy}(3): it is the
unique $j$ for which the mode numbers balance, $1+1+1=j$, so that the inner product attains its
maximal size).
\item For isolated gap frequencies, $\abs{(\omega^{(j)})^2-\tfrac{k^2}{1}(\omega^{(1)})^2}
\ge\tfrac{k^2}{2}(\omega^{(1)})^2=\Omega(k^2)$ for $k>1$.
\item For $1<j\ll n$, $(\omega^{(j)})^2\approx2k_2+\tfrac{k_2}{k_2-k_1}\big(\tfrac{j\pi}{n-1}\big)^2$,
so $(\omega^{(j)})^2-(\omega^{(1)})^2\approx\tfrac{k_1k_2(j^2-1)\pi^2}{2(k_2-k_1)(n-1)^2}$ and
$\ip{M(u^{(1)},u^{(1)},u^{(1)})}{u^{(j)}}/(\abs{u^{(1)}}^3\abs{u^{(j)}})\lesssim n^{-2}$, giving
$|c_{2,1,j}|=O(j^{-2})$.
\item For $n^{1-\eps}<j\le n_1-4$, the numerator is $\lesssim n^{2\eps}/n^2$ while the divisor is
$>k_1k_2 n^{2-2\eps}\pi^2/(4(k_2-k_1)n^2)$, giving $|c_{2,1,j}|=O(n^{-(2-4\eps)})$.
\item For in-band $j\ge n_1-3$ the numerator is $<(1+\eps_1)8n/n^2$ and the divisor exceeds
$\min\{2k_2-2k_1,(1-\eps_2)k_1\}$, giving $|c_{2,1,j}|=O(n^{-1})$.
\item For out-of-band $j$ the numerator is $<\tilde C/n^{3/2}$ and the divisor exceeds $\delta_2$,
giving $|c_{2,1,j}|=O(n^{-3/2})$.
\end{itemize}
The third-harmonic estimates follow identically, using the divisor
$(\omega^{(j)})^2-3^2(\omega^{(1)})^2$, which satisfies
$|(\omega^{(j)})^2-3^2(\omega^{(1)})^2|>\tfrac32(\omega^{(1)})^2=\Omega(1)$ for the near-edge
indices (the isolated frequencies lie in the gap, not above the band); together with the
numerator estimates this gives the $\Theta(1/n)$ size of $c_{2,3,1},c_{2,3,3}$. The relations
\eqref{eq:c21-estimates}--\eqref{eq:c23-estimates-rest} follow.
\end{proof}

Projecting \eqref{eq:order3} onto $\spanop\{u^{(1)}(e^{i\tau}+e^{-i\tau})\}$ gives the leading
frequency shift
\begin{equation}\label{eq:omega2-middle}
\omega^{\ord{2}}=-\frac{3\,\ip{M(u^{(1)},u^{(1)},u^{(1)})}{u^{(1)}}}
{2\abs{u^{(1)}}^4\omega^{\ord{0}}}
=-\frac{9}{\sqrt{2k_2}\,n}+\Theta\!\big(n^{-2}\big)=\Theta\!\big(n^{-1}\big)<0.
\end{equation}
Treating $\eps^2$ as the continuation variable, the frequency obeys
\begin{equation}\label{eq:omega-approx-middle}
\omega\approx\omega^{\ord{0}}+\eps^2\omega^{\ord{2}}
\approx\sqrt{2k_2}+\frac{k_1k_2}{4\sqrt{2k_2}(k_2-k_1)}\Big(\frac{\pi}{n-1}\Big)^2
-\frac{9}{\sqrt{2k_2}\,n}\,\eps^2 .
\end{equation}
Hence the frequency reaches the lower optical edge at
$\eps^2\approx\tfrac{k_1k_2}{36(k_2-k_1)}\tfrac{\pi^2}{n}=\Theta(1/n)$. Since
$|c_{2,1,3}|=\Theta(n)$ by \eqref{eq:c21-estimates}, the expansion \eqref{eq:expansion} ceases to
hold beyond this amplitude: the smallest divisor vanishes as the continued frequency meets the
band edge, exactly the threshold identified in Section~\ref{sec:intro}.

\begin{figure}[t]
\centering
\begin{subfigure}{0.46\textwidth}\centering
\includegraphics[width=\linewidth]{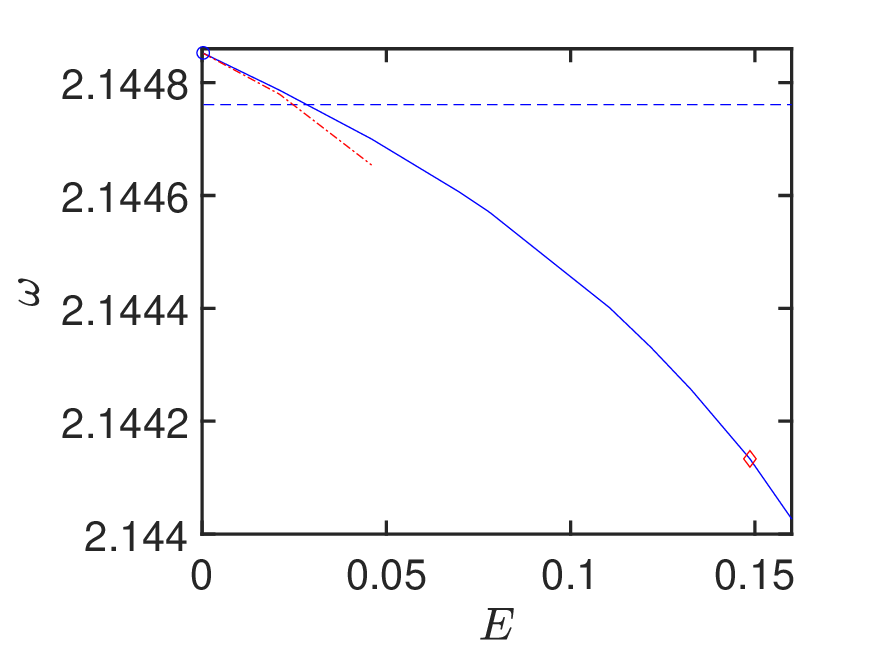}
\caption{}
\end{subfigure}\hfill
\begin{subfigure}{0.46\textwidth}\centering
\includegraphics[width=\linewidth]{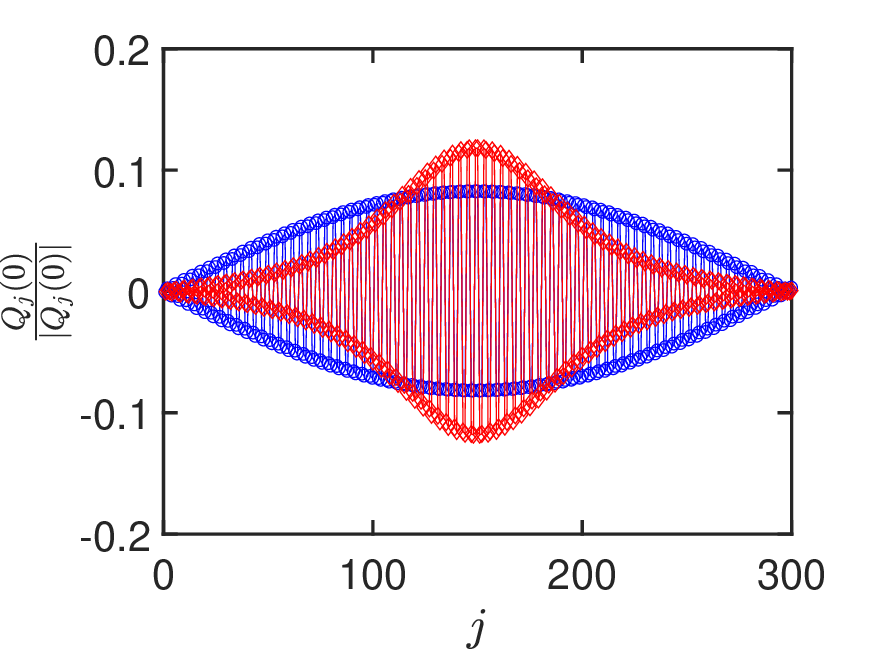}
\caption{}
\end{subfigure}
\caption{Regime \textup{(R1)} (middle-localized), $2n=200$, $k_1=1$, $k_2=2.3$,
$k_{3,1}=1.3$, $k_{3,2}=3.5$. (a) The branch frequency $\omega$ (solid) and its leading
approximation \eqref{eq:omega-approx-middle} (dash--dotted) as functions of the energy $E$, drift
downward from the lower optical edge $\sqrt{2k_2}$ (dashed horizontal). (b) The normalized
periodic orbit $Q(0)$ at small (blue) and large (red) amplitude: the uniformly extended near-edge
mode evolves into a state concentrated in the middle of the chain.}
\label{fig:middle}
\end{figure}

\subsection{Localization measure and the middle-localized profile}
To quantify localization we use, for a vector $f$ and an index window $[\ell_1,\ell_2]$,
\begin{equation}\label{eq:pdef}
p_{\ell_1,\ell_2}[f]=\frac{\bigl|\ip{I_{\ell_1,\ell_2}}{f^2}\bigr|}{\norm{f}^2},
\end{equation}
where $I_{\ell_1,\ell_2}$ has entries $1$ on $[\ell_1,\ell_2]$ and $0$ elsewhere; for a central
window $p$ increases as the orbit concentrates in the middle. Up to small $\eps$, $Q(0)$ is
governed by the in-phase superposition of $u^{(1)}$ and the resonant correction $u^{(3)}$, whose
coefficient $c_{2,1,3}=\Theta(n)$ drives the central concentration. The on-site definition of
$p_{m_1,2n+1-m_1}[\widetilde Q(0)]$ transfers verbatim, with the single change that the FPUT
$c_{2,1,3}$ equals $8$ times its on-site value (compare \eqref{eq:M-vs-hadamard}).

\begin{remark}[Closed-form approximation, regime R1]\label{rem:p-middle}
For odd $m_1\approx n/6$,
\begin{equation}\label{eq:p-middle}
p_{m_1,2n+1-m_1}[\widetilde Q(0)]\approx
\frac{(\tfrac56+\tfrac1{2\pi})+8n\eps^2\frac{3(k_2-k_1)(2-\sqrt3)}{32k_1k_2\pi^3}
+64n^2\eps^4(\tfrac56+\tfrac1{3\pi})\big[\frac{3(k_2-k_1)}{16k_1k_2\pi^2}\big]^2}
{1+64n^2\eps^4\big[\frac{3(k_2-k_1)}{16k_1k_2\pi^2}\big]^2}.
\end{equation}
This increases with $\eps$ for small $\eps$ ($n\eps^2\ll1$), confirming the growth of central
localization before the frequency leaves the band (Figure~\ref{fig:middle}).
\end{remark}

\section{The edge-localized family (regime R2)}\label{sec:boundary}

We now treat the second admissible boundary configuration
$\abs{k_{3,1}-k_2}\notapprox k_2$, $\abs{k_{3,2}-2k_2}\ll n^{-4}$, $k_{3,1}<2k_2$,
$k_{3,2}\le2k_2$. This regime tunes the right boundary spring to the optical edge, which changes
the near-edge spacing to the odd-multiple values $\Delta\theta^{(k)}\approx
(2k-1)\pi/(2(n-1))$ of \eqref{eq:dtheta} and moves the index of the dominant second-order
correction from the central index $j=3$ to the boundary index $j=2$. The analysis parallels
Section~\ref{sec:middle}; we record the modifications.

The reduction of Lemma~\ref{lem:M-vs-hadamard} persists, the only change being the index balance: in place of
$\sum_{i}(-1)^{j_i}a_i$ one tracks $2a_1-1+\sum_{i=2}^{4}(-1)^{j_i}(2a_i-1)$ and its vanishing on
$2\Z$. There is no essential difference.

\begin{lemma}[Bond map versus Hadamard cube, regime R2]\label{lem:M-vs-hadamard-R2}
Under $\abs{k_{3,1}-k_2}\notapprox k_2$ and $\abs{k_{3,2}-2k_2}\ll n^{-4}$, let
$1\le a_1,a_2,a_3\ll n$, $k=\max\{a_1,a_2,a_3\}$ (and $k=\max\{a_1,a_2,a_3,a_4\}$ if $a_4\ll n$).
Then for $a_4\ne n_1+1$,
\begin{equation}\label{eq:M-vs-hadamard-R2}
\Bigl|\ip{M(u^{(a_1)},u^{(a_2)},u^{(a_3)})}{u^{(a_4)}}
-8\,\ip{u^{(a_1)}\Hadamard u^{(a_2)}\Hadamard u^{(a_3)}}{u^{(a_4)}}\Bigr|
=O\!\Big(\frac{k^3}{n^2}\Big),\qquad a_4\ll n .
\end{equation}
Consequently Lemma~\ref{lem:M-dichotomy} holds verbatim in regime \textup{(R2)}. In particular,
if $a_1,a_2,a_3=\Theta(1)$ and $a_1+a_2+a_3<a_4\ll n$ (so that no choice of signs balances the
indices), then $\bigl|\ip{M(u^{(a)},u^{(b)},u^{(c)})}{u^{(d)}}\bigr|=O(1)$.
\end{lemma}

\begin{proposition}[Estimate of $Q^{\ord{2}}$, regime R2]\label{prop:Q2-R2}
In regime \textup{(R2)},
\begin{equation}\label{eq:Q2form-R2}
Q^{\ord{2}}(\tau)=\sum_{j=2}^{2n}c_{2,1,j}\frac{u^{(j)}}{\abs{u^{(j)}}}(e^{i\tau}+e^{-i\tau})
+\sum_{j=1}^{2n}c_{2,3,j}\frac{u^{(j)}}{\abs{u^{(j)}}}(e^{3i\tau}+e^{-3i\tau}),
\end{equation}
where now the resonant first-harmonic coefficient is at $j=2$:
\begin{equation}\label{eq:c21-R2}
\Bigl|c_{2,1,2}+\tfrac{6(k_2-k_1)n}{k_1k_2\pi^2}\Bigr|=O(1),
\end{equation}
and, for $j\ne2$,
\begin{equation}\label{eq:c21-R2-rest}
|c_{2,1,j}|=
\begin{cases}
O\!\big(j^{-2}\big), & j\in[2,n^{1-\eps}]\setminus\{2\},\\[1mm]
O\!\big(n^{-(2-4\eps)}\big), & j\in(n^{1-\eps},\,n_1-n^{1-\eps}],\\[1mm]
O\!\big(n^{-1}\big), & j>n_1-n^{1-\eps},\ (\omega^{(j)})^2\in(0,2k_1)\cup(2k_2,2k_1+2k_2),\\[1mm]
O\!\big(n^{-3/2}\big), & (\omega^{(j)})^2\in(2k_1,2k_2)\cup(2k_1+2k_2,+\infty),
\end{cases}
\end{equation}
and the third-harmonic coefficients satisfy
\begin{equation}\label{eq:c23-R2}
\Bigl|c_{2,3,1}+\tfrac{3}{8k_2 n}\Bigr|=O\!\big(n^{-2}\big),\qquad
\Bigl|c_{2,3,2}-\tfrac{1}{8k_2 n}\Bigr|=O\!\big(n^{-2}\big),
\end{equation}
together with, for the remaining indices,
\begin{equation}\label{eq:c23-R2-rest}
|c_{2,3,j}|=
\begin{cases}
O\!\big(n^{-1}\big), & (\omega^{(j)})^2\in(0,2k_1)\cup(2k_2,2k_1+2k_2),\\[1mm]
O\!\big(n^{-3/2}\big), & (\omega^{(j)})^2\in(2k_1,2k_2)\cup(2k_1+2k_2,+\infty).
\end{cases}
\end{equation}
\end{proposition}

\begin{proof}
Identical to Proposition~\ref{prop:Q2}, replacing the spacing \eqref{eq:dtheta} by its odd-multiple
version and the resonant index $3$ by $2$. The key inputs are
$\ip{M(u^{(1)},u^{(1)},u^{(1)})}{u^{(2)}}/(\abs{u^{(1)}}^3\abs{u^{(2)}})\approx-2n/n^2=-2/n$,
$\abs{u^{(1)}}^4\approx6n/n^2=6/n$, the near-edge dispersion
$(\omega^{(j)})^2\approx2k_2+\tfrac{k_2}{k_2-k_1}\big(\tfrac{(2j-1)\pi}{2(n-1)}\big)^2$, the gap
$(\omega^{(j)})^2-(\omega^{(1)})^2\approx\tfrac{k_1k_2((2j-1)^2-1)\pi^2}{2(k_2-k_1)4(n-1)^2}$ for
$2<j\ll n$, and the out-of-band decay bounds; substituting into \eqref{eq:c213-c233} yields
\eqref{eq:c21-R2}--\eqref{eq:c23-R2-rest}.
\end{proof}

Projecting onto $\spanop\{u^{(1)}(e^{i\tau}+e^{-i\tau})\}$ gives the same leading frequency shift
\begin{equation}\label{eq:omega2-R2}
\omega^{\ord{2}}=-\frac{3\,\ip{M(u^{(1)},u^{(1)},u^{(1)})}{u^{(1)}}}
{2\abs{u^{(1)}}^4\omega^{\ord{0}}}=-\frac{9}{\sqrt{2k_2}\,n}+\Theta\!\big(n^{-2}\big)
=\Theta\!\big(n^{-1}\big),
\end{equation}
so that
\begin{equation}\label{eq:omega-approx-R2}
\omega\approx\sqrt{2k_2}+\frac{k_1k_2}{16\sqrt{2k_2}(k_2-k_1)}\Big(\frac{\pi}{n-1}\Big)^2
-\frac{9}{\sqrt{2k_2}\,n}\,\eps^2 ,
\end{equation}
and the frequency crosses the lower optical edge at
$\eps^2\approx\tfrac{k_1k_2}{144(k_2-k_1)}\tfrac{\pi^2}{n}=\Theta(1/n)$. As before
$|c_{2,1,2}|=\Theta(n)$ once $\eps^2\gg1/n$, so the expansion fails beyond the threshold.

\begin{remark}[Closed-form approximation, regime R2]\label{rem:p-boundary}
With the boundary localization measure $p_{m_1,2n}[\widetilde Q(0)]$ defined as in
\eqref{eq:pdef} (right-end window), for even $m_1\approx n/3$,
\begin{equation}\label{eq:p-boundary}
p_{m_1,2n}[\widetilde Q(0)]\approx 1-
\frac{(\tfrac16-\tfrac1{2\pi})-8n\eps^2\frac{3(k_2-k_1)(2-\sqrt3)}{8k_1k_2\pi^3}
+64n^2\eps^4(\tfrac16-\tfrac1{3\pi})\big[\frac{3(k_2-k_1)}{4k_1k_2\pi^2}\big]^2}
{1+64n^2\eps^4\big[\frac{3(k_2-k_1)}{4k_1k_2\pi^2}\big]^2},
\end{equation}
which increases with $\eps$, confirming the growth of boundary localization before band crossing
(Figure~\ref{fig:boundary}).
\end{remark}

\begin{figure}[t]
\centering
\begin{subfigure}{0.46\textwidth}\centering
\includegraphics[width=\linewidth]{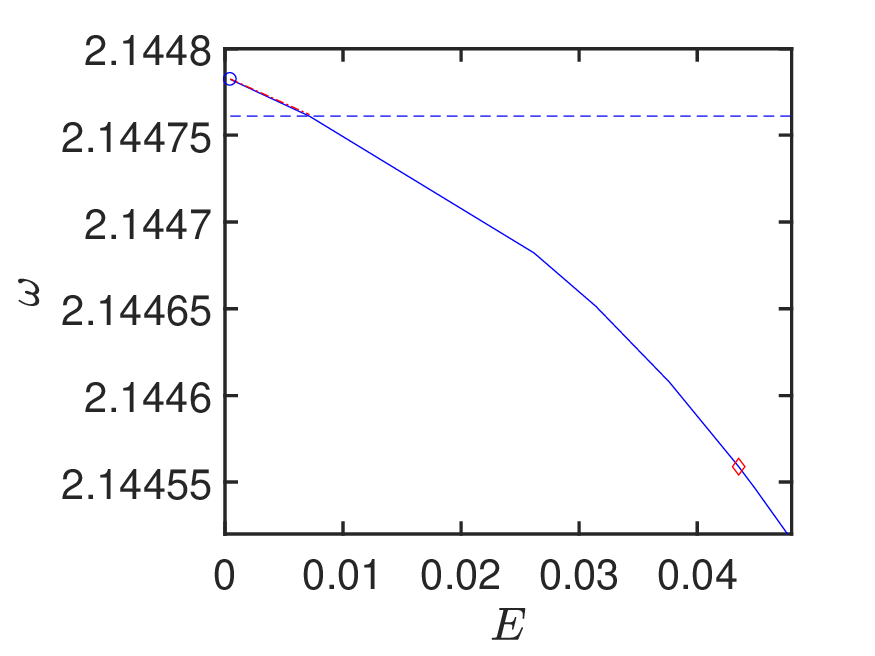}
\caption{}
\end{subfigure}\hfill
\begin{subfigure}{0.46\textwidth}\centering
\includegraphics[width=\linewidth]{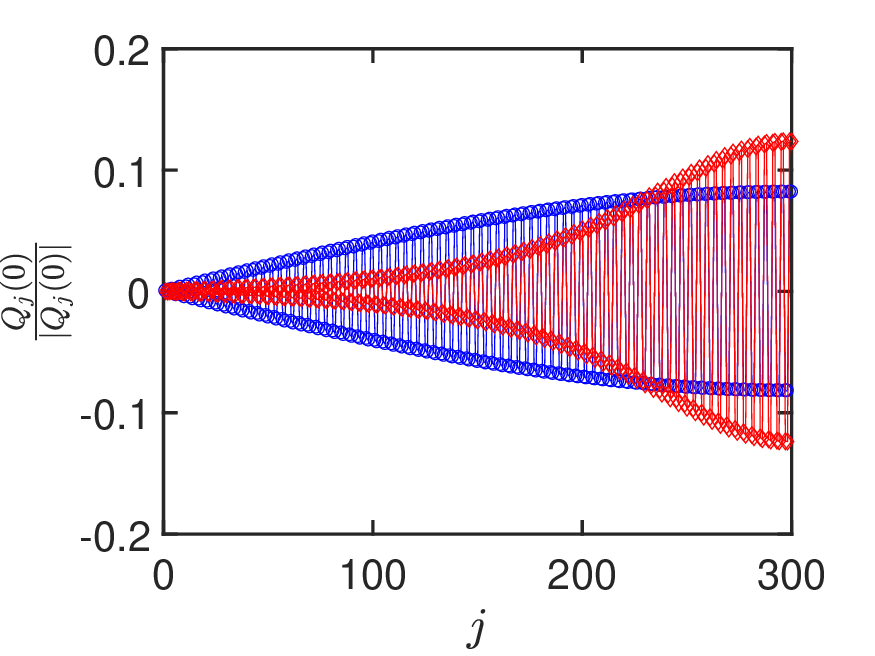}
\caption{}
\end{subfigure}
\caption{Regime \textup{(R2)} (edge localization), $2n=200$, $k_1=1$, $k_2=2.3$,
$k_{3,1}=1.3$, $k_{3,2}=4.6=2k_2$. (a) The branch frequency $\omega$ (solid) and its leading
approximation \eqref{eq:omega-approx-R2} (dash--dotted) drift down from the lower optical edge
(dashed horizontal). (b) The normalized periodic orbit $Q(0)$ at small (blue) and large (red)
amplitude: the uniformly extended mode evolves into a state concentrated at the right boundary.}
\label{fig:boundary}
\end{figure}

\section{Conclusion}\label{sec:conclusion}

We have shown that the finite-size framework introduced for on-site nonlinearity in
\cite{SongXuArxiv} carries over to the FPUT interaction nonlinearity, and that the transfer rests
on the reduction identity of Lemma~\ref{lem:M-vs-hadamard}, which equates the FPUT inner products
with the on-site Hadamard ones up to an $O(k^3/n^2)$ error. The main conclusions are the
following. First, the radius of convergence is a spectral quantity: it equals, up to constants,
the square root of the distance from the continued near-edge frequency to the band edge,
$\eps_{\mathrm{crit}}=\Theta(\sqrt{\mathrm{dist}((\omega^{(1)})^2,\partial\,\mathrm{band})})
=\Theta(1/\sqrt n)$, because the dominant perturbation coefficient at each order is amplified by
the same near-edge divisor of size $\Theta(1/n^2)$. This is meaningful only for a finite lattice,
where that distance is positive, and it is independent of which band edge is used as the anchor.
Second, a single interaction nonlinearity produces two distinct localized families, middle- and
edge-localized, depending only on the right boundary spring: the boundary fixes the near-edge
index spacing, which in turn selects, through the index-balance condition of
Lemma~\ref{lem:M-dichotomy}(3), whether the dominant second-order correction sits at the central
index ($j=3$, regime R1) or the boundary index ($j=2$, regime R2). This selection is set by the
discrete spectrum and is invisible to continuum or envelope reductions.

The diatomic chain is only a test case. The framework refers to no specific nonlinearity or band,
and the same strategy should apply to other odd short-range nonlinearities and, with a band edge
replaced by a curve or surface, to higher dimensions, with a convergence radius set by the same
$\Theta(1/n)$ shift and $\Theta(1/n^2)$ gap. A natural further question is the semi-infinite
chain. There the half-infinite operator has continuous spectrum and no isolated near-edge
frequency, so one would instead continue a genuine boundary-induced edge state, whose gap distance
is $\Theta(1)$ rather than $\Theta(1/n^2)$; a half-line Lyapunov--Schmidt or a spatial-dynamics
(centre-manifold) analysis would then test whether the same edge-detachment-and-localization
mechanism persists, now with an $O(1)$ convergence radius. We regard the principle that the
convergence radius of a finite-lattice perturbation series is a spectral quantity, to be estimated
at the same order as the gaps controlling it, as the main contribution of this work.

\section*{Acknowledgements}
The authors thank colleagues at Huazhong University of Science and Technology for helpful
discussions.

\appendix

\medskip\noindent
The appendices collect the proofs that underlie Theorem~\ref{thm:main}: the existence and
reduction of the periodic branch (Appendix~\ref{sec:proofs}), the reduction identity
(Appendix~\ref{app:reduction}), and the inductive convergence estimates
(Appendix~\ref{app:convergence}), the last organized around a unified weighted-sum machinery used
in both boundary regimes.

\section{Existence and reduction of the periodic branch}\label{sec:proofs}
Define, on $\Hper(\R/2\pi\Z\times\R^{2n})$,
\begin{equation}\label{eq:g-of-omega-Q}
g(\omega,Q(\tau))=\Big(\omega^2\frac{\dd^2}{\dd\tau^2}-L\Big)Q(\tau)-N(Q(\tau)),
\end{equation}
so that $g(\omega,0)=0$ and the linearization at the trivial branch is
$(\omega^{(1)})^2\tfrac{\dd^2}{\dd\tau^2}-L$, with kernel and cokernel \eqref{eq:kernel}. Let $P$ be
the projection of $\Hper$ onto $\spanop\{u^{(1)}e^{i\tau},u^{(1)}e^{-i\tau}\}$, $I$ the identity,
and split $\Hper=\spanop\{u^{(1)}e^{i\tau},u^{(1)}e^{-i\tau}\}\oplus Y$. Writing
\begin{equation}\label{eq:Q-split}
Q(\tau)=Ae^{i\tau}u^{(1)}+\bar A e^{-i\tau}u^{(1)}+v,\qquad v\in Y,
\end{equation}
equation \eqref{eq:fput-rescaled} is equivalent to the pair
\begin{align}
P\Big(\omega^2\tfrac{\dd^2}{\dd\tau^2}-L\Big)Q(\tau)&=PN(Q(\tau)),\label{eq:LSP}\\
(I-P)\Big(\omega^2\tfrac{\dd^2}{\dd\tau^2}-L\Big)Q(\tau)&=(I-P)N(Q(\tau)).\label{eq:LSQ}
\end{align}
For \eqref{eq:LSQ}, set
$g_1(\omega,A,\bar A,v)=\big(\omega^2\tfrac{\dd^2}{\dd\tau^2}-L\big)v-(I-P)N(Q(\tau))$. Then
$g_1(\omega^{(1)},0,0,0)=0$ and $\delta g_1/\delta v|_{(\omega^{(1)},0,0,0)}
=(\omega^{(1)})^2\tfrac{\dd^2}{\dd\tau^2}-L$ is a bijection onto $Y$; the implicit function
theorem yields $v=v(\omega,A,\bar A)$ with $v=\Theta(\abs{A}^2)$ near the origin. The invariances
of \eqref{eq:fput-rescaled} under $\tau\mapsto\tau+\eta$ and $\tau\mapsto-\tau$ give
\begin{equation}\label{eq:v-invariance}
v(\omega,A,\bar A)(\tau+\eta)=v(\omega,Ae^{i\eta},\bar A e^{-i\eta})(\tau),\qquad
v(\omega,A,\bar A)(-\tau)=v(\omega,\bar A,A)(\tau).
\end{equation}
The projected equation \eqref{eq:LSP} is equivalent to $g_2=g_3=0$, where
\begin{align}
g_2(\omega,A,\bar A,v)&=\tfrac12\!\int_0^{2\pi}\!\!\big(Ae^{i\tau}u^{(1)}+\bar A e^{-i\tau}u^{(1)}\big)\cdot
\Big[\big(\omega^2\tfrac{\dd^2}{\dd\tau^2}-L\big)Q-N(Q)\Big]\dd\tau,\label{eq:g2def}\\
g_3(\omega,A,\bar A,v)&=\tfrac1{2i}\!\int_0^{2\pi}\!\!\big(Ae^{i\tau}u^{(1)}-\bar A e^{-i\tau}u^{(1)}\big)\cdot
\Big[\big(\omega^2\tfrac{\dd^2}{\dd\tau^2}-L\big)Q-N(Q)\Big]\dd\tau.\label{eq:g3def}
\end{align}
The invariances \eqref{eq:v-invariance} imply
$g_2(\omega,Ae^{i\eta},\bar A e^{-i\eta},v)=g_2(\omega,A,\bar A,v)=g_2(\omega,\bar A,A,v)$, with the
analogous relations for $g_3$ and a sign change under reflection; taking $\eta=-2\arg A$ forces
$g_3\equiv0$. Restricting to time-even solutions and taking $\eta=-\arg A$ gives
$g_2(\omega,\abs{A},\abs{A},v)=:g_2(\omega,\abs A)$ with
$Q(\abs{A},v,\tau)=2\abs{A}u^{(1)}\cos\tau+v(\abs{A},\omega)$. Writing the strains of this $Q$ at
each site and cubing them, the FPUT nonlinearity \eqref{eq:Nform} expands into the explicit
trigonometric form
\begin{equation}\label{eq:N-expanded}
\begin{aligned}
N\big(Q(\abs{A},\omega,\tau)\big)
=\ &8\abs{A}^3\Big[(u^{(1)}-\Sm u^{(1)})^{\Hadamard3}-(\Sp u^{(1)}-u^{(1)})^{\Hadamard3}\Big]\cos^3\tau\\
&+12\abs{A}^2\Big[(u^{(1)}-\Sm u^{(1)})^{\Hadamard2}\Hadamard(v-\Sm v)
-(\Sp u^{(1)}-u^{(1)})^{\Hadamard2}\Hadamard(\Sp v-v)\Big]\cos^2\tau\\
&+6\abs{A}\Big[(u^{(1)}-\Sm u^{(1)})\Hadamard(v-\Sm v)^{\Hadamard2}
-(\Sp u^{(1)}-u^{(1)})\Hadamard(\Sp v-v)^{\Hadamard2}\Big]\cos\tau\\
&+\Big[(v-\Sm v)^{\Hadamard3}-(\Sp v-v)^{\Hadamard3}\Big],
\end{aligned}
\end{equation}
in which the first term, through $\cos^3\tau=\tfrac34\cos\tau+\tfrac14\cos3\tau$, supplies the
resonant first-harmonic forcing that drives the reduction. Using
$g_4(\omega,\abs{A})=g_2(\omega,\abs{A})/\abs{A}^2$ and $v=\Theta(\abs{A}^2)$, only the
$8\abs A^3\cos^3\tau$ term in \eqref{eq:N-expanded} contributes at the bifurcation point, so
$g_4(\omega^{(1)},0)=0$ and
\begin{equation}\label{eq:g4-deriv}
\frac{\partial g_4}{\partial\omega}\Big|_{(\omega,\abs A)=(\omega^{(1)},0)}
=-4\omega^{(1)}\abs{u^{(1)}}^2\!\int_0^{2\pi}\!\cos^2\tau\,\dd\tau
=-4\pi\,\omega^{(1)}\abs{u^{(1)}}^2\ne0 .
\end{equation}
The implicit function theorem then yields $\omega=\omega(\abs A)$ and $v=v(\omega(\abs A),\abs A)$
in a neighbourhood of $(\omega^{(1)},0)$, completing the proof. \qed

\section{The reduction identity}\label{app:reduction}
\subsection{Componentwise structure of the trilinear form $M$}
By the parametrization \eqref{eq:eigform}, for an eigenmode $u^{(\ell)}$ with $\ell\ll n$ the two
spring strains are
\begin{align}
(\Sp u^{(\ell)}-u^{(\ell)})_{2j-1}=(u^{(\ell)}-\Sm u^{(\ell)})_{2j}
&=u^{(\ell)}_{2j}-u^{(\ell)}_{2j-1}\nonumber\\
&=(-1)^{j}2\sin\Delta\alpha^{(\ell)}\cos\!\big(\Delta\beta^{(\ell)}-(j-1)\Delta\theta^{(\ell)}\big),
\label{eq:strain-odd}\\
(\Sp u^{(\ell)}-u^{(\ell)})_{2j}=(u^{(\ell)}-\Sm u^{(\ell)})_{2j+1}
&=u^{(\ell)}_{2j+1}-u^{(\ell)}_{2j}\nonumber\\
&=(-1)^{j}\big[\sin(\Delta\alpha^{(\ell)}+\Delta\beta^{(\ell)}-j\Delta\theta^{(\ell)})\nonumber\\
&\qquad\qquad+\sin(-\Delta\alpha^{(\ell)}+\Delta\beta^{(\ell)}-(j-1)\Delta\theta^{(\ell)})\big].\label{eq:strain-even}
\end{align}
In particular $|u^{(\ell)}_{2j}-u^{(\ell)}_{2j-1}|=O(\ell/n)$, while the even-site strains are
$\Theta(1)$ for the near-edge modes used here, the quantitative form of Remark~\ref{rem:edge}.
Define the strain remainder $D^{\pm}_{\ell,j}$ by
\begin{equation}\label{eq:Ddef}
D^{\pm}_{\ell,j}=(u^{(\ell)}_{2j+1}-u^{(\ell)}_{2j})
-(-1)^{j}2\sin\!\big(\pm\Delta\alpha^{(\ell)}+\Delta\beta^{(\ell)}-(j-\tfrac12\pm\tfrac12)\Delta\theta^{(\ell)}\big),
\end{equation}
so that, expanding the trigonometric factors,
\begin{equation}\label{eq:Dform}
\begin{aligned}
D^{\pm}_{\ell,j}=(-1)^j\Big[&\sin(\Delta\beta^{(\ell)}-j\Delta\theta^{(\ell)})
\big(\pm(\cos(\Delta\theta^{(\ell)}-\Delta\alpha^{(\ell)})-\cos\Delta\alpha^{(\ell)})\big)\\
&+\cos(\Delta\beta^{(\ell)}-j\Delta\theta^{(\ell)})
\big(\pm(\sin(\Delta\theta^{(\ell)}-\Delta\alpha^{(\ell)})-\sin\Delta\alpha^{(\ell)})\big)\Big],
\end{aligned}
\end{equation}
and one checks directly that $D^{+}_{\ell,j}=-D^{-}_{\ell,j}$.

\subsection{Proof of Lemma~\ref{lem:M-vs-hadamard}}
Write $G(a_1,a_2,a_3)$ for the set of cyclic permutations of $(a_1,a_2,a_3)$,
\[
G(a_1,a_2,a_3)=\{(a_1,a_2,a_3),\,(a_2,a_3,a_1),\,(a_3,a_1,a_2)\}.
\]
The odd components of $M(u^{(a_1)},u^{(a_2)},u^{(a_3)})$ are, for $1\le j\le n-1$,
\begin{equation}\label{eq:Modd}
\begin{aligned}
M(u^{(a_1)},u^{(a_2)},u^{(a_3)})_{2j+1}
&=8(-1)^{j}\prod_{i=1}^{3}\sin(\Delta\alpha^{(a_i)}+\Delta\beta^{(a_i)}-j\Delta\theta^{(a_i)})\\
&\quad+4\!\!\sum_{(a,b,c)\in G}\!\!D^{+}_{a,j}\,\sin(\Delta\alpha^{(b)}+\Delta\beta^{(b)}-j\Delta\theta^{(b)})
\sin(\Delta\alpha^{(c)}+\Delta\beta^{(c)}-j\Delta\theta^{(c)})\\
&\quad+2(-1)^{j}\!\!\sum_{(a,b,c)\in G}\!\!D^{+}_{a,j}D^{+}_{b,j}\,
\sin(\Delta\alpha^{(c)}+\Delta\beta^{(c)}-j\Delta\theta^{(c)})\\
&\quad+D^{+}_{a_1,j}D^{+}_{a_2,j}D^{+}_{a_3,j}
+8(-1)^{j}\prod_{i=1}^{3}\sin\Delta\alpha^{(a_i)}\cos(\Delta\beta^{(a_i)}-j\Delta\theta^{(a_i)}),
\end{aligned}
\end{equation}
with the analogous expression for the even components $M(\cdots)_{2j}$ obtained from
\eqref{eq:strain-odd}--\eqref{eq:Dform}, and the boundary components
\begin{align}
M(\cdots)_{1}&=8\prod_{i=1}^{3}\sin(\Delta\alpha^{(a_i)}+\Delta\beta^{(a_i)})
-8\prod_{i=1}^{3}\sin\Delta\alpha^{(a_i)}\cos\Delta\beta^{(a_i)}
-7\prod_{i=1}^{3}\sin(\Delta\alpha^{(a_i)}+\Delta\beta^{(a_i)}),\label{eq:M1}\\
M(\cdots)_{2n}&=(-1)^{n-1}8\prod_{i=1}^{3}\sin(\Delta\beta^{(a_i)}-\Delta\alpha^{(a_i)}-(n-1)\Delta\theta^{(a_i)})\nonumber\\
&\quad+(-1)^{n}8\prod_{i=1}^{3}\sin\Delta\alpha^{(a_i)}\cos(\Delta\beta^{(a_i)}-(n-1)\Delta\theta^{(a_i)})\nonumber\\
&\quad+(-1)^{n}7\prod_{i=1}^{3}\sin(\Delta\beta^{(a_i)}-\Delta\alpha^{(a_i)}-(n-1)\Delta\theta^{(a_i)}).\label{eq:M2n}
\end{align}
Using $(n-1)\Delta\theta^{(a_i)}=a_i\pi+\widetilde\theta^{(a_i)}$ from \eqref{eq:dtheta}, the
boundary terms obey $\prod_i\sin(\Delta\beta^{(a_i)}-\Delta\alpha^{(a_i)}-(n-1)\Delta\theta^{(a_i)})
=O(k^3/n^3)$ and $\prod_i\sin\Delta\alpha^{(a_i)}=O(k^3/n^3)$, hence
$|M(\cdots)_{2n}|=O(k^3/n^3)$.

Now decompose the difference
\begin{equation}\label{eq:Kdecomp}
\ip{M(u^{(a_1)},u^{(a_2)},u^{(a_3)})}{u^{(a_4)}}
-8\,\ip{u^{(a_1)}\Hadamard u^{(a_2)}\Hadamard u^{(a_3)}}{u^{(a_4)}}=K_1+K_2+K_3+K_4+K_5,
\end{equation}
where $K_1$ collects the products of one strain remainder $D$ with two leading sines, $K_2$ the
products of two remainders with one sine, $K_3$ the triple-remainder term
$\sum_j\big(u^{(a_4)}_{2j+1}+u^{(a_4)}_{2j}\big)D^{+}_{a_1,j}D^{+}_{a_2,j}D^{+}_{a_3,j}$, $K_4$ the
boundary-strain $\sin\Delta\alpha$ corrections, and $K_5$ the endpoint terms
\eqref{eq:M1}--\eqref{eq:M2n}. Since
$|u^{(a_4)}_1|,|u^{(a_4)}_{2n}|\le1$ and $\Delta\alpha^{(a_i)}=\Theta(a_i/n)=\Delta\beta^{(a_i)}$,
the endpoint term satisfies
\begin{equation}\label{eq:K5}
|K_5|=O\!\big(k^4/n^4\big)\ (a_4\ll n),\qquad |K_5|=O\!\big(k^3/n^3\big)\ (a_4\in(n^{1-\eps},n_1]).
\end{equation}
For $K_1$, write $\ell_4=a_4$ and group the summands over $G(a_1,a_2,a_3)$. Telescoping the sine
products through the product-to-sum identity and using
\[
\prod_{i=2}^{4}\sin\!\big(\Delta\alpha^{(\ell_i)}-\tfrac12\Delta\theta^{(\ell_i)}\big)
\big[\cos(\Delta\theta^{(\ell_1)}-\Delta\alpha^{(\ell_1)})-\cos\Delta\alpha^{(\ell_1)}\big]=O(k^4/n^4)
\]
together with the bound $\big|\sum_{j=1}^{n-1}\sin(\Delta\beta^{(\ell_1)}-j\Delta\theta^{(\ell_1)})
\prod_{i=2}^4\cos(\cdots)\big|\le n-1$ gives, after using the symmetry in $a_1,a_2,a_3$,
$K_1=O(k^3/n^2)$. For $K_2$ one splits $K_2=K_{21}+K_{22}$; using
$(\sin(\Delta\theta-\Delta\alpha)-\sin\Delta\alpha)(\cos(\Delta\theta-\Delta\alpha)-\cos\Delta\alpha)
=O(k^3/n^3)$ and $\prod_{i=3}^4[\,\cdots]=O(k/n)$ gives $K_{21}=O(k^4/n^3)$, while
$\prod_{i=1}^2[\cos(\Delta\theta-\Delta\alpha)-\cos\Delta\alpha]=O(k^4/n^4)$ gives
$K_{22}=O(k^5/n^4)$, whence $K_2=O(k^4/n^3)$. For $K_3$, substituting \eqref{eq:Dform} and
simplifying to $K_3=\sum_{j}(u^{(a_4)}_{2j+1}+u^{(a_4)}_{2j})D^{+}_{a_1,j}D^{+}_{a_2,j}D^{+}_{a_3,j}
=\sum_{i=1}^{4}K_{3i}$ with
\[
|K_{31}|=O(k^6/n^5),\quad |K_{32}|=O(k^3/n^2),\quad |K_{33}|=O(k^5/n^4),\quad |K_{34}|=O(k^4/n^3),
\]
using $\sin(\Delta\theta^{(\ell_i)}-\Delta\alpha^{(\ell_i)})-\sin\Delta\alpha^{(\ell_i)}=O(k/n)$ and
$\cos(\Delta\theta^{(\ell_i)}-\Delta\alpha^{(\ell_i)})-\cos\Delta\alpha^{(\ell_i)}=O(k^2/n^2)$, so
$K_3=O(k^3/n^2)$. For $K_4$, $\prod_{i=1}^3\sin\Delta\alpha^{(a_i)}=O(k^3/n^3)$ yields
$K_4=O(k^3/n^2)$. Adding the five estimates establishes \eqref{eq:M-vs-hadamard}. \qed

\begin{corollary}\label{cor:K-O1}
If $a_i=\Theta(1)$ for $i=1,2,3,4$, then $K_1=O(n^{-2})$, $K_2=O(n^{-3})$, $K_3=O(n^{-2})$,
$K_4=O(n^{-2})$, $K_5=O(n^{-4})$.
\end{corollary}

\section{Convergence estimates}\label{app:convergence}
\subsection{Unified weighted-sum machinery}
We now set up the inductive estimates that yield the convergence radius. Define, for $m\ge0$,
the index set $S_m=\{(j,l):1\le j\le2n,\,0\le l\le m\}$, and the sub-windows
\begin{equation}\label{eq:Ddomains}
D_1=\{2\le j\le n_1\},\quad D_2=\{2n-n_2+1\le j\le2n\},\quad
D_3=\begin{cases}\{n_1+1,\,2n-n_2\}&\text{(R1)},\\ \{n_1+1\}&\text{(R2)},\end{cases}
\end{equation}
together with the partition $S_m=\bigcup_{s=0}^{3}F_{s,m}$,
\[
\begin{gathered}
F_{0,m}=\{(j,l):j=1,\,1\le l\le m\},\qquad
F_{1,m}=\{(j,l):j\in D_1,\,0\le l\le m\}\cup\{(1,0)\},\\
F_{s,m}=\{(j,l):j\in D_s,\,0\le l\le m\},\quad s=2,3.
\end{gathered}
\]
We also use the following index families, as in \cite{SongXuArxiv}:
\begin{equation}\label{eq:Esets}
\begin{aligned}
\mathbb{E}_1^0(k)&=\{(x,y,z)\in\N^3:x+y+z=k\},\\
\mathbb{E}_2^0(k)&=\{(x,y,z)\in\N^3:x+y+z=k+1\},\\
\mathbb{E}_2^2(k)&=\mathbb{E}_2^0(k)\setminus\{(0,0,k+1)\},\\
\mathbb{E}_1^{1*}(k)&=\{(x,y,z)\in(\N^{*})^3:x+y+z=k,\ x\le y\le z\},\\
\mathbb{E}_2^{1*}(k)&=\{(x,y,z)\in(\N^{*})^3:x+y+z=k+1,\ x\le y\},\\
\mathbb{E}_3^{0*}(k)&=\{(x,y)\in(\N^{*})^2:x+y=k+1\},\\
\mathbb{E}_3^{1*}(k)&=\{(x,y)\in(\N^{*})^2:x+y=k+1,\ x\le y\}.
\end{aligned}
\end{equation}
The following arithmetic inequalities are the combinatorial backbone.

\begin{proposition}[Weighted-sum inequalities, {\cite{SongXuArxiv}}]\label{prop:arith}
The following hold:
\begin{align}
\sum_{(a,b)\in \mathbb{E}_3^{1*}(k)}\frac{1}{(a^2+2)(b^2+2)}&<\frac{2\pi^2-8}{3}\cdot\frac{2}{(k+1)^2+2}
<\frac{4}{(k+1)^2+2},\label{eq:arith1}\\
\sum_{(a,b,c)\in \mathbb{E}_2^{1*}(k)}\frac{1}{(a^2+2)(b^2+2)(c^2+2)}
&<\Big(\frac{2\pi^2-8}{3}\Big)\Big(\frac{4\pi^2-4}{3}\Big)\frac{1}{(k+1)^2+2}<\frac{48}{(k+1)^2+2},\label{eq:arith2}\\
\sum_{(a,b,c)\in \mathbb{E}_1^{1*}(k)}\frac{1}{(a^2+2)(b^2+2)(c^2+2)}
&<\Big(\frac{2\pi^2-2}{3}\Big)^2\frac{1}{k^2+2}<\frac{36}{k^2+2}.\label{eq:arith3}
\end{align}
\end{proposition}

These inequalities are purely arithmetic, not involving the lattice or the nonlinearity, and
were established in our on-site work; we therefore quote them and refer to
\cite{SongXuArxiv} for the proof. (The constants come from
$\sum_{a\ge0}1/(a^2+2)<(2\pi^2-2)/3$ and $\sum_{a\ge1}1/(a^2+2)<(2\pi^2-8)/3$, applied to the
index sets $\mathbb{E}_3^{1*}$, $\mathbb{E}_2^{1*}$, $\mathbb{E}_1^{1*}$ defined above.)

For multi-indices we define the product symbol
\begin{equation}\label{eq:Pdef}
P(a_1,a_2,a_3,l_1,l_2,l_3,j_1,j_2,j_3)
=M\!\big(u^{(j_1)},u^{(j_2)},u^{(j_3)}\big)\prod_{i=1}^{3}\frac{c_{2a_i,2l_i+1,j_i}}{\abs{u^{(j_i)}}},
\end{equation}
and the weighted sum
\begin{equation}\label{eq:Wdef}
\begin{aligned}
\mathcal{W}(a_1,a_2,a_3,s_1,s_2,s_3,j)
=\!\!\sum_{(j_1,l_1)\in F_{s_1,a_1}}\sum_{(j_2,l_2)\in F_{s_2,a_2}}\sum_{(j_3,l_3)\in F_{s_3,a_3}}
&\frac{\bigl|\ip{u^{(j_1)}u^{(j_2)}u^{(j_3)}}{u^{(j)}}\bigr|}
{\abs{u^{(j_1)}}\abs{u^{(j_2)}}\abs{u^{(j_3)}}\abs{u^{(j)}}}\\
&\times\prod_{i=1}^{3}\abs{c_{2a_i,2l_i+1,j_i}} .
\end{aligned}
\end{equation}

\begin{proposition}\label{prop:W}
Set $R_m=C_1 f^{m}n^{m}/(m^2+2)$ with $f=640(k_2-k_1)/(k_1k_2)$ and $C_1>0$, and assume the
inductive bounds \textup{(i)--(vii)} of Lemma~\ref{lem:induction} hold up to order $k$. Then:
\begin{enumerate}[label=\textup{(\arabic*)},leftmargin=1.6em]
\item For $j\notin D_3$ and all three $s_i=1$, the weighted sum obeys, according to how many of
$a_1,a_2,a_3$ vanish,
\begin{equation}\label{eq:Wcases}
\mathcal{W}(a_1,a_2,a_3,1,1,1,j)<(1+\eps)\,\frac{1}{n}\times
\begin{cases}
32\,C_3\,R_{a_3}, & a_1=a_2=0,\ a_3\ne0,\\[1mm]
32\,C_3^2\,R_{a_2}R_{a_3}, & a_1=0,\ a_2a_3\ne0,\\[1mm]
32\,C_3^3\,R_{a_1}R_{a_2}R_{a_3}, & a_1a_2a_3\ne0.
\end{cases}
\end{equation}
In all remaining cases (some $s_i\ne1$, or $j\in D_3$) the same sum is $\ll R_{a_1}R_{a_2}R_{a_3}/n$,
with the convention $R_0=1$.
\item Consequently, using Proposition~\ref{prop:arith},
\begin{equation}\label{eq:Wmaster}
\begin{aligned}
\sum_{(a_1,a_2,a_3)\in \mathbb{E}_1^0(k)}\ \sum_{(j_i,l_i)\in S_{a_i}}
&\frac{\bigl|\ip{P(a_1,a_2,a_3,l_1,l_2,l_3,j_1,j_2,j_3)}{u^{(j)}}\bigr|}{\abs{u^{(j)}}}\\
&\qquad<\frac{R_{k+1}}{n^2}\cdot\frac{12\cdot16\,C_3}{f}\,(1+8C_1C_3+32C_1^2C_3^2)(1+\eps).
\end{aligned}
\end{equation}
\end{enumerate}
\end{proposition}

\begin{proof}
The bounds in (1) follow from Lemma~\ref{lem:M-dichotomy}: the generic ($\notin D_3$, all
$s_i=1$) case sits in the resonant branch (3) with value $\Theta(n)$, dividing by the norms
$\abs{u}=\Theta(\sqrt n)$ and the inductive coefficient bounds $\sum|c|<C_3R$; the off-diagonal
cases land in the $O(1)$ or $O(k^3/n^2)$ branches and are $\ll$. Summing over
$\mathbb{E}_1^0(k)$ and applying the arithmetic inequalities \eqref{eq:arith1}--\eqref{eq:arith3} (with
the weight $R_aR_bR_c\sim f^{a+b+c}n^{a+b+c}/\prod(a_i^2+2)$ and $a+b+c=k$) collapses the triple
sum to a single $R_{k+1}/n^2$ term with the stated constant. This is \eqref{eq:Wmaster}.
\end{proof}

\subsection{The inductive convergence estimate}
We now state the master induction, whose hypotheses close the recursion and whose conclusion
(i) furnishes the frequency bound and (ii)--(vii) the coefficient bounds.

\begin{lemma}[Master induction]\label{lem:induction}
In either regime, with $R_m=C_1 f^{m}n^{m}/(m^2+2)$, $f=640(k_2-k_1)/(k_1k_2)$, $C_1>0$, there
exist $C_2,C_3,C_4>0$ with $C_1C_3=\tfrac1{128}$ and $C_1C_2=\tfrac{k_1k_2}{20\sqrt{2k_2}(k_2-k_1)}$
such that for all $m\ge1$ and all admissible harmonic indices $l$:
\begin{align}
&\textup{(i)}\ \ |\omega^{\ord{2m}}|<\frac{C_2R_m}{n^2};
&&\textup{(ii)}\ \ \sum_{j\in D_1}|c_{2m,1,j}|<C_3R_m; \label{eq:ind-i-ii}\\
&\textup{(iii)}\ \ \sum_{j\in D_2}|c_{2m,1,j}|<\frac{C_4R_m}{n};
&&\textup{(iv)}\ \ \sum_{j\in D_3}|c_{2m,1,j}|<\frac{C_4R_m}{n^{5/2}}; \label{eq:ind-iii-iv}\\
&\textup{(v)}\ \ \sum_{j=1}^{n_1}|c_{2m,2l+1,j}|<\frac{C_4R_m}{n(2l+1)^2};
&&\textup{(vi)}\ \ \sum_{j\in D_2}|c_{2m,2l+1,j}|<\frac{C_4R_m}{n(2l+1)^2}; \label{eq:ind-v-vi}\\
&\textup{(vii)}\ \ \sum_{j\in D_3}|c_{2m,2l+1,j}|<\frac{C_4R_m}{n^{5/2}(2l+1)^2}. \label{eq:ind-vii}
\end{align}
\end{lemma}

\begin{proof}
\emph{Base case $m=1$.} From \eqref{eq:omega2-middle} (resp.\ \eqref{eq:omega2-R2}),
$|\omega^{\ord{2}}|\approx9/(\sqrt{2k_2}n)<C_2R_1/n^2$ since $R_1=C_1fn/3$ and
$C_1C_2=k_1k_2/(20\sqrt{2k_2}(k_2-k_1))$. From Proposition~\ref{prop:Q2} (resp.\
\ref{prop:Q2-R2}), $\sum_{j\in D_1}|c_{2,1,j}|\approx 3(k_2-k_1)n/(2k_1k_2\pi^2)<C_3R_1$, and the
out-of-window and higher-harmonic bounds (iii)--(vii) hold for $C_4$ chosen large (still
$\Theta(1)$). This establishes the base case.

\emph{Induction step.} Assume (i)--(vii) for $m=1,\dots,k$. Collecting the $\Theta(\eps^{2k+3})$
terms in \eqref{eq:fput-rescaled},
\begin{equation}\label{eq:order2k3}
\begin{aligned}
\Big((\omega^{\ord{0}})^2\frac{\dd^2}{\dd\tau^2}-L\Big)Q^{\ord{2k+2}}(\tau)
&+\!\!\sum_{(r,s,t)\in \mathbb{E}_2^2(k)}\!\!\omega^{\ord{2r}}\omega^{\ord{2s}}
\frac{\dd^2}{\dd\tau^2}Q^{\ord{2t}}(\tau)\\
&=\!\!\sum_{(a_1,a_2,a_3)\in \mathbb{E}_1^0(k)}\!\!M\!\big(Q^{\ord{2a_1}},Q^{\ord{2a_2}},Q^{\ord{2a_3}}\big).
\end{aligned}
\end{equation}

\emph{(i): frequency bound.} Projecting onto $\spanop\{u^{(1)}(e^{i\tau}+e^{-i\tau})\}$,
\begin{equation}\label{eq:omega-proj}
\begin{aligned}
2|\omega^{\ord{0}}\omega^{\ord{2k+2}}|
\le\!\!\sum_{(a_1,a_2,a_3)\in \mathbb{E}_1^0(k)}\sum_{(j_i,l_i)\in S_{a_i}}
&\frac{4\bigl|\ip{P(a_1,a_2,a_3,l_1,l_2,l_3,j_1,j_2,j_3)}{u^{(1)}}\bigr|}{\abs{u^{(1)}}}\\
&+\!\!\sum_{(r,s)\in \mathbb{E}_3^{0*}(k)}\!\!|\omega^{\ord{2r}}\omega^{\ord{2s}}|.
\end{aligned}
\end{equation}
The first sum is bounded by Proposition~\ref{prop:W}\,(2). The second is controlled by the
inductive (i):
\begin{equation}\label{eq:omega-omega}
\sum_{(r,s)\in \mathbb{E}_3^{0*}(k)}|\omega^{\ord{2r}}\omega^{\ord{2s}}|
\le\sum_{(r,s)\in \mathbb{E}_3^{1*}(k)}\frac{2C_1^2C_2^2 f^{k+1}n^{k-3}}{(r^2+2)(s^2+2)}
<\frac{32C_1C_2^2 R_{k+1}}{n^4}\ll\frac{R_{k+1}}{n^2}.
\end{equation}
Combining, and using $\omega^{\ord{0}}=\sqrt{2k_2}+O(1/n^2)$, the relations $C_1C_3=1/128$
and $C_1C_2=k_1k_2/(20\sqrt{2k_2}(k_2-k_1))$ give
\begin{equation}\label{eq:i-close}
|\omega^{\ord{2k+2}}|
<\frac{C_2R_{k+1}}{n^2}\cdot\frac{24\cdot16\,C_1C_3}{(C_1C_2)\,f\sqrt{2k_2}}\,\frac{545}{512}(1+\eps)
<\frac{C_2R_{k+1}}{n^2},
\end{equation}
which closes (i) at order $k+1$ (the numerical constant $\tfrac{3\cdot545}{32\cdot512}<1$).

\emph{(ii): first-harmonic sum on $D_1$.} Projecting \eqref{eq:order2k3} onto
$\spanop\{u^{(j)}(e^{i\tau}+e^{-i\tau})\}$ and summing over $j\in D_1$,
\begin{equation}\label{eq:ii-proj}
\begin{aligned}
\sum_{j\in D_1}|c_{2k+2,1,j}|
&<4\!\!\sum_{(a_1,a_2,a_3)\in \mathbb{E}_1^0(k)}\sum_{s_1,s_2,s_3}\sum_{j\in D_1}
\frac{\mathcal{W}(a_1,a_2,a_3,s_1,s_2,s_3,j)}{|(\omega^{(j)})^2-(\omega^{(1)})^2|}\\
&\quad+\!\!\sum_{(r,s,t)\in \mathbb{E}_2^2(k)}\!\!\frac{|\omega^{\ord{2r}}\omega^{\ord{2s}}|
\sum_{j\in D_1}|c_{2t,1,j}|}{|(\omega^{(j)})^2-(\omega^{(1)})^2|}.
\end{aligned}
\end{equation}
For the first term, the divisor estimates
$(\omega^{(j)})^2-(\omega^{(1)})^2>k_1k_2\pi^2/(4(k_2-k_1)n^{2\eps})$ for $2\le j\le n^{1-\eps}$ and
$(\omega^{(j)})^2-(\omega^{(1)})^2\approx k_1k_2(j^2-1)\pi^2/(2(k_2-k_1)(n-1)^2)$ for
$n^{1-\eps}<j<n_1$ combine with Proposition~\ref{prop:W} to give a bound
$<\tfrac{\pi^2}{6}\tfrac{2(k_2-k_1)}{k_1k_2\pi^2}R_{k+1}\tfrac{12\cdot16C_3}{f}(1+8C_1C_3+32C_1^2C_3^2)(1+\eps)$.
The second term is bounded, via (i) and the inductive (ii), by
$<R_{k+1}\,\tfrac{32\sqrt{2k_2}C_1C_2C_3(k_2-k_1)}{k_1k_2\pi^2}(1+\eps)$. Adding,
\begin{equation}\label{eq:ii-close}
\sum_{j\in D_1}|c_{2k+2,1,j}|
<C_3 R_{k+1}\Big(\tfrac{8}{5\pi^2}+\tfrac{2\cdot545}{5\cdot512}\Big)(1+\eps)<C_3R_{k+1},
\end{equation}
closing (ii).

\emph{(iii): first-harmonic sum on $D_2$.} Projecting and summing over $j\in D_2$, the divisor is
now $\Theta(n)$ (the right-boundary block), so
\begin{equation}\label{eq:iii-proj}
\sum_{j\in D_2}|c_{2k+2,1,j}|
<4\!\!\sum_{\mathbb{E}_1^0(k)}\sum_{s_i}\sum_{j\in D_2}\frac{\mathcal{W}}{|(\omega^{(j)})^2-(\omega^{(1)})^2|}
+\!\!\sum_{\mathbb{E}_2^2(k)}\frac{|\omega^{\ord{2r}}\omega^{\ord{2s}}|\sum_{j\in D_2}|c_{2t,1,j}|}
{|(\omega^{(j)})^2-(\omega^{(1)})^2|}.
\end{equation}
Using $|(\omega^{(j)})^2-(\omega^{(1)})^2|\gtrsim k_2-k_1$ on $D_2$ together with
Proposition~\ref{prop:W} and inductive (iii), both terms are $\ll R_{k+1}/n$, so for $C_4$ large
$\sum_{j\in D_2}|c_{2k+2,1,j}|<C_4R_{k+1}/n$, closing (iii).

\emph{(iv)--(vii).} These are estimated exactly as (iii): the windows $D_3$ and the higher
harmonics $l\ge1$ carry the additional small factors $n^{-3/2}$ and $(2l+1)^{-2}$ respectively,
which are inherited from the divisor $(\omega^{(j)})^2-(2l+1)^2(\omega^{(1)})^2=\Theta((2l+1)^2)$
and the out-of-window decay; choosing $C_4$ sufficiently large (independent of $n$) closes them.
This completes the induction.
\end{proof}

\subsection{Proof of Theorem~\ref{thm:main}}
By Lemma~\ref{lem:induction}, the norm of the even-order correction is controlled by the
coefficient sums. Writing $S_m=\{(j,l):j,l\ge0\}$ for the active index set at order $m$,
\begin{equation}\label{eq:Qnorm}
\begin{aligned}
\norm{Q^{\ord{2m}}(\tau)}^2=\frac1\pi\sum_{(j,l)\in S_m}|c_{2m,2l+1,j}|^2
<\frac1\pi\Big(&\sum_{\lambda=1}^{3}\sum_{j\in D_\lambda}|c_{2m,1,j}|
+\sum_{l=1}^{m}\sum_{j=1}^{n_1}|c_{2m,2l+1,j}|\\
&+\sum_{\lambda=2}^{3}\sum_{l=1}^{m}\sum_{j\in D_\lambda}|c_{2m,2l+1,j}|\Big)^2.
\end{aligned}
\end{equation}
Inserting bounds (ii)--(vii) and summing the convergent series $\sum_l(2l+1)^{-2}$ yields
\begin{equation}\label{eq:Qnorm-bound}
\norm{Q^{\ord{2m}}(\tau)}^2<C_3^2R_m^2 .
\end{equation}
Since $R_m=C_1 f^m n^m/(m^2+2)$,
\begin{equation}\label{eq:ratio}
\lim_{m\to\infty}\frac{R_{m+1}}{R_m}=fn,
\end{equation}
so the series \eqref{eq:expansion} for $Q$ converges for $\eps<1/\sqrt{fn}$. The frequency bound
(i) gives the same radius for the $\omega$-series. This proves convergence on $0<\eps^2<1/(fn)$,
i.e.\ $\eps=\Theta(1/\sqrt n)$, with $f=640(k_2-k_1)/(k_1k_2)$ in both regimes (the
regime-dependent prefactors enter only the second-order coefficients of
Sections~\ref{sec:middle}--\ref{sec:boundary}). The constant $f$ is deliberately conservative: it
comes from the uniform bound $\abs{\ip{M}{\cdot}}\le32n$ used in Lemma~\ref{lem:M-dichotomy}(1) in
place of the asymptotically tight $16n=8\cdot2n$, so that the lower-order $O(k^3/n^2)$ correction in
\eqref{eq:M-vs-hadamard} is absorbed uniformly without being tracked; retaining that correction
would sharpen $f$ to $320(k_2-k_1)/(k_1k_2)$ and enlarge the radius by a factor $\sqrt2$, which does
not affect the order $\eps=\Theta(1/\sqrt n)$. The monotone downward frequency drift is given by
\eqref{eq:omega2-middle} for regime R1 and by \eqref{eq:omega2-R2} for regime R2, and the
localization growth before band crossing is described in Remarks~\ref{rem:p-middle}
and~\ref{rem:p-boundary}. \qed



\begin{thebibliography}{99}\small



\bibitem{MacKayAubry} R.~S.~MacKay, S.~Aubry, \emph{Proof of existence of breathers for
time-reversible or Hamiltonian networks of weakly coupled oscillators}, Nonlinearity
\textbf{7} (1994), 1623--1643.




\bibitem{FlachGorbach} S.~Flach, A.~V.~Gorbach, \emph{Discrete breathers: advances in theory and
applications}, Phys.\ Rep.\ \textbf{467} (2008), 1--116.



\bibitem{JamesNoble2004} G.~James, P.~Noble, \emph{Breathers on diatomic Fermi--Pasta--Ulam
lattices}, Physica D \textbf{196} (2004), 124--171.





\bibitem{Vainchtein2022} A.~Vainchtein, \emph{Solitary waves in FPU-type lattices}, Physica D
\textbf{434} (2022), 133252.










\bibitem{HasanKane} M.~Z.~Hasan, C.~L.~Kane, \emph{Colloquium: Topological insulators}, Rev.\
Mod.\ Phys.\ \textbf{82} (2010), 3045--3067.

\bibitem{SSH} W.~P.~Su, J.~R.~Schrieffer, A.~J.~Heeger, \emph{Solitons in polyacetylene}, Phys.\
Rev.\ Lett.\ \textbf{42} (1979), 1698--1701.

\bibitem{WangBertoldi} P.~Wang, L.~Lu, K.~Bertoldi, \emph{Topological phononic crystal with
one-way elastic edge waves}, Phys.\ Rev.\ Lett.\ \textbf{115} (2015), 104302.

\bibitem{SmirnovaNonlinear} D.~Smirnova, D.~Leykam, Y.~Chong, Y.~Kivshar, \emph{Nonlinear
topological photonics}, Appl.\ Phys.\ Rev.\ \textbf{7} (2020), 021306.

\bibitem{ChaunsaliStrong} R.~Chaunsali, H.~Xu, J.~Yang, P.~G.~Kevrekidis, G.~Theocharis,
\emph{Stability of topological edge states under strong nonlinear effects}, Phys.\ Rev.\ B
\textbf{103} (2021), 024106.




\bibitem{LeykamChong} D.~Leykam, Y.~D.~Chong, \emph{Edge solitons in nonlinear-photonic
topological insulators}, Phys.\ Rev.\ Lett.\ \textbf{117} (2016), 143901.















\bibitem{HadadKhanikaev} Y.~Hadad, A.~B.~Khanikaev, A.~Al\`u, \emph{Self-induced topological
transitions and edge states supported by nonlinear staggered potentials}, Phys.\ Rev.\ B
\textbf{93} (2016), 155112.


\bibitem{SongXuJPA} H.~Song, H.~Xu, \emph{Analytical estimations of edge states and extended
states in large finite-size lattices}, J.\ Phys.\ A: Math.\ Theor.\ \textbf{59} (2026), 095701.




\bibitem{SongXuArxiv} H.~Song, H.~Xu, \emph{Emergence of purely nonlinear localized states with
frequencies exited from spectral bands}, arXiv:2303.05247.




\bibitem{Kielhofer} H.~Kielh\"ofer, \emph{Bifurcation Theory: An Introduction with Applications to
Partial Differential Equations}, 2nd ed., Springer, New York, 2012.














\bibitem{Kuznetsov} Y.~A.~Kuznetsov, \emph{Elements of Applied Bifurcation Theory}, 4th ed.,
Springer, Cham, 2023.




\bibitem{FPUT} E.~Fermi, J.~Pasta, S.~Ulam, \emph{Studies of nonlinear problems. I.}, in
\emph{The Many-Body Problem} (D.~C.~Mattis, ed.), World Scientific, Singapore, 1993.




\bibitem{HornJohnson} R.~A.~Horn, C.~R.~Johnson, \emph{Matrix Analysis}, 2nd ed., Cambridge
University Press, Cambridge, 2012.




\bibitem{Parlett} B.~N.~Parlett, \emph{The Symmetric Eigenvalue Problem}, SIAM, Philadelphia,
1998.








\end{thebibliography}
\end{document}